\documentclass[ejsv2]{imsart}

\RequirePackage[numbers]{natbib}
\RequirePackage[colorlinks,citecolor=blue,urlcolor=blue]{hyperref}
\RequirePackage{graphicx}

\startlocaldefs
\theoremstyle{plain}

\newtheorem{theorem}{Theorem}[section]

\theoremstyle{definition}

\theoremstyle{remark}

\theoremstyle{definition}
\newtheorem{assumption}[theorem]{Assumption}

\usepackage{amsmath,amsfonts,mathtools}
\usepackage{mathrsfs} 
\usepackage{dsfont}
\usepackage{svg}
\usepackage{caption,subcaption}
\usepackage{algorithm,algorithmic}
\usepackage{enumitem}
\usepackage{xcolor,todonotes}
\usepackage{ulem}
\usepackage{relsize}
\newcommand{\R}{\mathds R}
\newcommand{\E}{\mathds E}

\newcommand{\I}{\mathbf{I}}

\newcommand{\F}{\mathcal{F}}

\newcommand{\Tr}{\text{Tr}\,}
\newcommand{\valmax}{\lambda_{\text{max}}}

\endlocaldefs

\begin{document}
\begin{frontmatter}
\title{An hybrid stochastic Newton algorithm for logistic regression}
%\title{A sample article title with some additional note\thanksref{t1}}
\runtitle{Hybrid Stochastic Newton algorithm}
%\thankstext{T1}{A sample additional note to the title.}

\begin{aug}
%%%%%%%%%%%%%%%%%%%%%%%%%%%%%%%%%%%%%%%%%%%%%%%
%% Only one address is permitted per author. %%
%% Only division, organization and e-mail is %%
%% included in the address.                  %%
%% Additional information can be included in %%
\author[A]{\fnms{Bernard}~\snm{Bercu}\ead[label=e1]{bernard.bercu@math.u-bordeaux.fr}},
\author[A]{\fnms{Luis}~\snm{Fredes}\ead[label=e2]{luis.fredes@math.u-bordeaux.fr}}
\and
\author[A]{\fnms{Eméric}~\snm{Gbaguidi}\ead[label=e3]{thierry-emeric.gbaguidi@math.u-bordeaux.fr}}
%%%%%%%%%%%%%%%%%%%%%%%%%%%%%%%%%%%%%%%%%%%%%%
%% Addresses                                %%
%%%%%%%%%%%%%%%%%%%%%%%%%%%%%%%%%%%%%%%%%%%%%%
\address[A]{Institut de Mathématiques de Bordeaux, Université de Bordeaux \printead[presep={,\ \\ }]{e1,e2}\\\printead{e3}}

%\address[B]{\printead[]{e3}}
\runauthor{Bercu, Fredes and Gbaguidi}
\end{aug}

\begin{abstract}
In this paper, we investigate a second-order stochastic algorithm for solving large-scale binary classification problems. We propose to make use of a new hybrid stochastic Newton algorithm that includes two weighted components in the Hessian matrix estimation: the first one coming from the natural Hessian estimate and the second associated with the stochastic gradient information. Our motivation comes from the fact that both parts evaluated at the true parameter of logistic regression, are equal to the Hessian matrix. This new formulation has several advantages and it enables us to prove the almost sure convergence of our stochastic algorithm to the true parameter. Moreover, we significantly improve the almost sure rate of convergence to the Hessian matrix. Furthermore, we establish the central limit theorem for our hybrid stochastic Newton algorithm. Finally, we show a surprising result on the almost sure convergence of the cumulative excess risk.
\end{abstract}

\begin{keyword}[class=MSC]
\kwd{49M15, 68W27, 62J12, 60B10, 60F05}
%\kwd[Primary ]{49M15, 68W27, 62J12, 60B10, 60F05}
%\kwd{00X00}
%\kwd[; secondary ]{(to be completed)}
\end{keyword}

\begin{keyword}
\kwd{Stochastic optimization; stochastic Newton algorithm; logistic regression; almost sure convergence; central limit theorem}
\end{keyword}

\end{frontmatter}
%%%%%%%%%%%%%%%%%%%%%%%%%%%%%%%%%%%%%%%%%%%%%%
%% Please use \tableofcontents for articles %%
%% with 50 pages and more                   %%
%%%%%%%%%%%%%%%%%%%%%%%%%%%%%%%%%%%%%%%%%%%%%%
%\tableofcontents

\section{Introduction}

Our goal is to investigate the standard stochastic optimization problem in $\R^d$
\begin{equation}\tag{$\mathcal{P}$}\label{problem}
    \min_{h\in \R^d} G(h),
\end{equation}
where 
\begin{equation}
    G(h)=\E[g(\Phi,Y,h)].
\end{equation}
In what follows, we focus our attention on the well-known logistic regression \cite{bach2010self,bercu2020efficient,boyer2023asymptotic}. For that purpose, we introduce a sequence $(\Phi_n,Y_{n})$ of random vectors taking values in $\R^d \times \{0, 1\}$. Furthermore, we assume that $(\Phi_n)$ is a sequence of independent and identically distributed random vectors such that for all $n\geqslant 1$, the conditional distribution of $Y_{n}$ is a Bernoulli distribution given by
\[\mathcal{L}(Y_{n}\lvert \Phi_n )=\mathcal{B}(\pi(\theta^T\Phi_n)), \qquad \text{where} \qquad \pi(x)=\dfrac{e^x}{1+e^x},\]
and $\theta$ is the unknown parameter belonging to $\R^d$ of the logistic regression. Denote by $G$ the convex function given, for all $h\in \R^d$, by
\begin{equation}
    G(h)=\E\left[\log\Big(1+\exp(h^T \Phi)\Big)- h^T \Phi Y\right],
\end{equation}
where $\mathcal{L}(Y\lvert \Phi)=\mathcal{B}(\pi(\theta^T\Phi))$ and $\Phi$ shares the same distribution as the sequence $(\Phi_n)$.
The gradient and Hessian matrix of $G$ are respectively given by 
\begin{equation}\label{general_grad_logistic}
    \nabla G(h)=\E\left[(\pi(h^T \Phi)-Y)\Phi\right],
\end{equation}
and
\begin{equation}\label{general_hessian_logistic}
    \nabla^2 G(h)=\E\Big[\pi(h^T \Phi)\big(1-\pi(h^T \Phi)\big)\Phi \Phi^T \Big].
\end{equation}
Hence, our goal is to find the unique solution of \eqref{problem} which satisfies $ \nabla G(\theta)=0$.
One can observe that $\E[Y\lvert \Phi]=\pi(\theta^T \Phi)$ and 
\begin{align}
	\Sigma(\theta)=\E\left[\Big(\pi(\theta^T \Phi)-Y\Big)^2\Phi \Phi^T\right]&=\E\left[\E\Big[\Big(\pi(\theta^T \Phi)-Y\Big)^2 \Big \lvert \Phi \Big]\Phi \Phi^T\right] \nonumber \\
	&=\E\Big[\pi(\theta^T \Phi)\big(1-\pi(\theta^T \Phi)\big)\Phi \Phi^T \Big] \nonumber \\
	&=\nabla^2 G(\theta).\label{general_sigma_logistic}
\end{align}
This simple calculation highlights our motivation to propose our hybrid stochastic Newton algorithm defined in Section \ref{sec:framework}.

Several methods have been investigated in order to find $\theta$ such as the stochastic gradient descent and second-order algorithms \cite{bercu2020efficient,hazan2007logarithmic,pelletier1998weak,polyak1992acceleration,robbins1951stochastic}. Here, we shall focus our attention on the family of stochastic Newton algorithms. These methods are more able to deal with the problem where the Hessian matrix of the objective function has eigenvalues with significantly different values. A truncated version of the standard stochastic Newton algorithm has been proposed in \cite{bercu2020efficient}. This truncation was necessary to ensure the convergence of the algorithm but it significantly reduces the rate of convergence to the Hessian matrix. We also refer the reader to \cite{boyer2023asymptotic,de2021stochastic,godichon2024online} for other references on truncated stochastic Newton algorithms. In this paper, we will avoid any truncation in our hybrid stochastic Newton algorithm.  We highlight that in \cite{bercu2020efficient}, the inverse of the Hessian matrix estimation is updated recursively by using the Sherman-Morrison-Woodbury inversion formula \cite{horn2012matrix}. This tool avoids the inversion of large matrix at each iteration.

More recently, a weighted averaged version of the stochastic Newton algorithm was proposed in \cite{boyer2023asymptotic} for more general problems including the case of linear, logistic, and softmax regressions. The almost sure convergence of the iterates to the optimum and their rates of convergence as well as a central limit theorem have been established in \cite{boyer2023asymptotic}. As mentioned before, the same truncated version of stochastic Newton algorithm as \cite{bercu2020efficient} has been maintained in \cite{boyer2023asymptotic} for the specific case of logistic regression. Godichon-Baggioni et al. \cite{godichon2024online}
focused on the recursive estimation technique for the inverse of the Hessian matrix using a Robbins-Monro procedure and considered a weighted averaged version to further enhance the rates of convergence. However, we will significantly improve their rates with our hybrid stochastic Newton  algorithm.

Bach \cite{bach2010self} borrowed tools from self-concordance to analyze the statistical properties of logistic regression estimate and established some oracle inequalities on the excess risk based on concentration inequalities for U-statistics.  For their part, De Vilmarest and Wintenberger \cite{de2021stochastic} provided explicit bounds with high probability on this convergence for the same truncated stochastic Newton algorithm named there the static extended Kalman filter in the logistic regression case. They derived local bounds on the cumulative excess risk with high probability via a matrix concentration inequality. However, the result are proved under a strong convergence assumption which is that with high probability, there exists a convergence time after which the algorithm stays trapped in a local region around the optimum. 

Another variants of the stochastic Newton algorithm have been investigated to address large-scale applications \cite{byrd2016stochastic,mokhtari2014res,shanno1970conditioning,ye2017approximate}. For example, Kovalev et al. \cite{kovalev2019stochastic} proposed the stochastic cubic Newton algorithm which combines cubic regularization with the stochastic Newton technique for minimizing the average of a very large number of sufficiently smooth and strongly convex functions. Furthermore, in the deterministic case, Hazan \cite{hazan2016introduction} introduced the online Newton step which is also known as a quasi-Newton method or Gauss-Newton algorithm \cite{bercu2023stochastic,cenac2025efficient,schraudolph2007stochastic}. However, strictly speaking, this algorithm is first order, in the sense that it only uses gradient information. Next, Hazan \cite{hazan2016introduction} provided the logarithmic regret guarantee of the online Newton step algorithm for exp-concave functions. Recently, Wintenberger \cite{wintenberger2024stochastic} showed that the stochastic online Newton step algorithm achieves fast-rate stochastic regrets under the stochastically exp-concave assumption. 
Moreover, Yousefian et al. \cite{yousefian2020stochastic} studied an iterative regularized stochastic quasi-Newton method for the regularized problems with nonstrongly convex objective functions and Lipschitz continuous gradient. They showed the almost sure convergence of this scheme towards the optimal value of the problem under the bounded variance hypothesis.

In what follows, we highlight our main contributions. On the one hand, we introduce our new hybrid stochastic Newton algorithm motivated by \eqref{general_sigma_logistic}. Our parametric approach incorporates two components in the Hessian matrix estimation: one from the natural Hessian estimate and the second from the stochastic gradient information. This helps to establish many convergence guarantees while avoiding the useless introduction of truncation as it was done in the previous literature. On the other hand, beyond the results of almost sure convergence, we have significantly improved the rate of convergence to the Hessian matrix compared to the current state of art. Lastly, we establish the central limit theorem for our hybrid stochastic Newton algorithm.

The rest of the paper is organized as follows. In Section \ref{sec:framework}, we describe our iterative hybrid stochastic Newton algorithm. Section \ref{sec:main_results} deals with our main results. We prove the almost sure convergence of our algorithm in Theorem \ref{hsna_th_convps}. Hereafter, we provide in Theorem \ref{hsna_th_rates}, the almost sure rates of convergence for the estimation of the parameter $\theta$ as well as for the Hessian matrix $\nabla^2 G(\theta)$. Moreover, Theorem \ref{hsna_th_clt} states the central limit theorem for our hybrid stochastic Newton algorithm. Furthermore, we present the simulations performed on a large-scale classification problem with both synthetic and real datasets in Section \ref{sec:experiments}. Finally, the paper ends with all technical proofs in the appendices.

\section{Our hybrid stochastic Newton algorithm}\label{sec:framework}
In this section, we introduce our new hybrid stochastic Newton algorithm. Next, we will discuss the assumptions required for analysis. Our parametric approach provides a unified framework for the stochastic Newton algorithm and the online Newton step. Our hybrid stochastic Newton algorithm is given, for all $n\geqslant 0$, by 
\begin{align}
    a_{n+1} &= \pi(\widehat{\theta}_{n}^T\Phi_{n+1})(1-\pi(\widehat\theta_n^T\Phi_{n+1})),\nonumber\\
    b_{n+1} &= \left(\pi(\widehat\theta_{n}^T\Phi_{n+1}) - Y_{n+1}\right)^2,\nonumber\\
    c_{n+1} &= \alpha a_{n+1} + \beta b_{n+1},\label{hsna_c_n}\\
    S_{n+1}^{-1} &= S_{n}^{-1} -c_{n+1}(1+c_{n+1}\Phi ^T_{n+1}S_{n}^{-1}\Phi_{n+1})^{-1}S_{n}^{-1}\Phi_{n+1}\Phi_{n+1}^TS_{n}^{-1},\label{hsna_S_n_inv}\\
    \widehat\theta_{n+1} &= \widehat\theta_n- S_{n+1}^{-1}\Phi_{n+1}\left( \pi(\widehat\theta_{n}^T\Phi_{n+1}) - Y_{n+1} \right) \label{hsna},
\end{align}
where $\alpha$ and $\beta$ are two positive constants such that $\beta>0$ and the initial value $\widehat{\theta}_{0}$ is a bounded vector of $\R^d$ which can be arbitrarily chosen.
We take $S_0 =\I_d$ and observe that for all $n\geqslant 0$,
\begin{equation}\label{hsna_S_n}
    S_{n+1}=\sum_{k=1}^{n+1} c_k \Phi_k \Phi_k^T + \I_d= S_n + c_{n+1}\Phi_{n+1}\Phi_{n+1}^T.
\end{equation}
The first key idea underlying the hybrid stochastic Newton algorithm is that we incorporate two weighted components in the Hessian matrix estimation: one from the natural Hessian estimate and the second from the stochastic gradient information. The motivation comes from \eqref{general_sigma_logistic} which shows that both parts evaluated in $\theta$ are equal to the Hessian matrix $\nabla^2 G(\theta)$. This interesting relation is specific to the logistic regression problem. The remarkable advantage of this new formulation is also that we do not need anymore truncation to obtain the convergence guarantees.

Moreover, the another cornerstone in our new algorithm is that we update $S_{n+1}$ and $S_{n+1}^{-1}$ before updating $\widehat\theta_{n+1}$. Our motivation is to make use of the maximum information available for the estimation of the Hessian matrix $\nabla^2 G( \theta)$ at the current time $n+1$.  We wish to point out that in all the previous literature \cite{bercu2020efficient,boyer2023asymptotic,godichon2024online}, $\widehat\theta_{n+1}$ is updated before $S_{n+1}$ and $S_{n+1}^{-1}$. Hence, this is the second key point in our analysis and we will see below the advantages of that choice. In fact, our approach makes a significant difference in the convergence analysis, but raises other problems of measurability. However, we find how to properly overcome these limitations.

In the same vein as previous literature on the Newton algorithm \cite{bercu2020efficient,de2021stochastic,hazan2016introduction}, the inverse of the Hessian estimate is updated via the useful Sherman-Morrison-Woodbury inversion formula (see e.g. \cite[p.~19]{horn2012matrix}) applied here in \eqref{hsna_S_n_inv}. We recall that the classical analysis of the stochastic Newton methods is based on the study of the eigenvalues convergence of the Hessian matrix estimation. This will not be the case here at all and we adopt another strategy. Therefore, we make extensive use of this recursive form \eqref{hsna_S_n_inv} from which we establish new interesting equations for the convergence of estimates. Finally, notice that our hybrid stochastic Newton algorithm with $\alpha=0$ and $\beta=1$, leads to the online Newton step.
Furthermore, we consider the positive definite matrices defined, for all $n\geqslant 1$, by
\begin{equation*}
    H_n=\sum_{k=1}^n a_k \Phi_k \Phi_k^T \qquad \text{and} \qquad
    \Sigma_n=\sum_{k=1}^n b_k \Phi_k \Phi_k^T,
\end{equation*}
as well as their averaged versions given by,
\begin{equation*}
    \overline{H}_n=\dfrac{1}{n}H_n, \qquad 
    \overline{\Sigma}_n=\dfrac{1}{n}\Sigma_n \qquad \text{and} \qquad
    \overline{S}_n=\dfrac{1}{n}S_n.
\end{equation*}
%$\G_{n} = \sigma\left( (\Phi_1,Y_1),(\Phi_2,Y_2),\dots,(\Phi_n,Y_n), \Phi_{n+1}\right) = \sigma(\mathcal{F}_n,\Phi_{n+1})$
We will prove that these three matrices $\overline{H}_n$, $\overline{\Sigma}_n$ and $\overline{S}_n$ converge almost surely to the Hessian matrix $\nabla^2 G(\theta)$. Throughout the paper, we will use the filtration $(\F_n)$ given, for all $n\geqslant 1$, by $\F_{n} = \sigma\left( (\Phi_1,Y_1),(\Phi_2,Y_2),\dots,(\Phi_n,Y_n)\right)$.
That choice is a quite of importance and we will see in our proofs how to overcome the measurability problems due to $S_{n+1}$ in the algorithm update \eqref{hsna}. Hereafter, our main assumptions are as follows.
\begin{assumption}\label{sna_cond1}
The Hessian matrix $\nabla^2 G(\theta)$ and the matrix $\E[\Phi\Phi^T]$ are positive definite.
\end{assumption}
\begin{assumption}\label{sna_cond2}
It exists a positive constant $d_\Phi$ such that $\lVert \Phi\rVert \leqslant d_\Phi$ almost surely.
\end{assumption}
Assumption \ref{sna_cond1} is a classical regularity hypothesis in the literature \cite{bercu2020efficient,boyer2023asymptotic}. Assumption \ref{sna_cond2} requires that the vector $\Phi$ is bounded. This condition seems rather restrictive, but it is made in most of all papers dealing with stochastic Newton algorithm for logistic regression \cite{bach2010self,de2021stochastic,wintenberger2024stochastic} and is satisfied in statistical applications.

\section{Main results}\label{sec:main_results}
This section is devoted to the main results of the paper. On the one hand, we establish the almost sure convergence of our hybrid stochastic Newton algorithm in Theorem \ref{hsna_th_convps}. We next present the almost sure rates of convergence in Theorem \ref{hsna_th_rates} and discuss their improvement in view of the current state of the art. On the other hand, we deal with the central limit theorem of our hybrid stochastic Newton algorithm in Theorem \ref{hsna_th_clt}.
\begin{theorem}\label{hsna_th_convps}
Assume that Assumptions \ref{sna_cond1} and \ref{sna_cond2} hold. Then, we have the following almost sure convergences
\begin{align}
    \lim_{n\to +\infty} \widehat{\theta}_n&=\theta \qquad \text{a.s.,}\label{hsna_th_convps_res1}\\
    \lim_{n\to +\infty} \overline{H}_n=\lim_{n\to +\infty}& \overline{\Sigma}_n=\nabla^2 G(\theta) \qquad \text{a.s.,} \label{hsna_th_convps_res2}\\
    \lim_{n\to +\infty} \overline{S}_n&=(\alpha+\beta)\nabla^2 G(\theta) \qquad \text{a.s.} \label{hsna_th_convps_res3}
\end{align}
\end{theorem}
\begin{proof}
The proof of Theorem \ref{hsna_th_convps} can be found in Appendix \ref{app_hsna_th_convps}.
\end{proof}

Hereafter, we state the almost sure rates convergence of our hybrid stochastic Newton algorithm.
\begin{theorem}\label{hsna_th_rates}
Assume that Assumptions \ref{sna_cond1} and \ref{sna_cond2} are satisfied. Then, we have the almost sure rate of convergence,
\begin{equation}\label{hsna_th_rates_res1}
    \lVert \widehat{\theta}_n-\theta \rVert^2 =\mathcal{O}\left( \dfrac{\log(n)}{n} \right) \qquad \text{a.s.}
\end{equation}
Moreover, we also have
\begin{equation}\label{hsna_th_rates_res2}
    \lVert \overline{S}_{n} -(\alpha+\beta)\nabla^2 G(\theta) \rVert^2= \mathcal{O}\left( \dfrac{\log(n)}{n} \right) \qquad \text{a.s.,}
\end{equation}
and 
\begin{equation} \label{hsna_th_rates_res3}
    \left\lVert \overline{S}_{n}^{-1} -(\alpha+\beta)^{-1}\left(\nabla^2 G(\theta)\right)^{-1} \right\rVert^2= \mathcal{O}\left( \dfrac{\log(n)}{n} \right) \qquad \text{a.s.}
\end{equation}
\end{theorem}
\begin{proof}
The proof of Theorem \ref{hsna_th_rates} can be found in Appendix \ref{app_hsna_th_rates}.
\end{proof}
We point out that our Theorem \ref{hsna_th_rates} significantly improves the rate of convergence to the Hessian matrix found in Theorem 4.3 of \cite{bercu2020efficient} which was of order $\mathcal{O}(1/n^\gamma)$ with $0<\gamma<1$ . The reason is that we do not use any truncation which involves a weighting slowing down the method. Therefore, our algorithm leads to a better estimation of the Hessian matrix.
Lastly, the central limit theorem for our hybrid stochastic Newton algorithm is as follows.
\begin{theorem}\label{hsna_th_clt}
Assume that Assumptions \ref{sna_cond1} and \ref{sna_cond2} are satisfied and that $\alpha+\beta=1$ with $\beta>0$. Then, we have the central limit theorem
\begin{equation}\label{hsna_th_clt_res1}
    \sqrt{n} (\widehat{\theta}_n-\theta ) \quad \overset{\mathcal{L}}{\underset{n\to +\infty}{\longrightarrow}} \quad \mathcal{N}_d\left(0,\left(\nabla^2 G(\theta)\right)^{-1}\right).
\end{equation}
\end{theorem}
\begin{proof}
The proof of Theorem \ref{hsna_th_clt} can be found in Appendix \ref{app_hsna_th_clt}.
\end{proof}
Theorem \ref{hsna_th_clt} states that the asymptotic covariance matrix is the inverse of the Hessian of the function $G$ evaluated at $\theta$. Hence, this result remains consistent with the expected one which is that stochastic Newton methods are asymptotically efficient. Thus, the two weightings with $\alpha$ and $\beta$, introduced in the Hessian matrix estimation $S_{n+1}$ does not affect this efficiency once again thanks to the useful property $\Sigma (\theta) = \nabla^2 G(\theta)$ for the logistic regression. In addition, we deduce from the central limit theorem \eqref{hsna_th_clt_res1} and a careful analysis of the remainder term in Taylor expansion of $G(\widehat{\theta}_n)$, see \eqref{hsna_th_qsl_eq11b}, that
\begin{equation}
    n \Big(G(\widehat{\theta}_n)-G(\theta)\Big) \quad \overset{\mathcal{L}}{\underset{n\to +\infty}{\longrightarrow}} \quad \dfrac{1}{2}\chi^2(d),
\end{equation}
where $\chi^2(d)$ stands for the chi-squared distribution with $d$ degrees of freedom.
\par
Finally, we propose a quadratic strong law for our hybrid stochastic Newton which implies a surprising result on the almost sure convergence of the cumulative excess risk. They require no more assumptions than those we have already made and are stated as follows.
\begin{theorem}\label{hsna_th_qsl}
Assume that Assumptions \ref{sna_cond1} and \ref{sna_cond2} are satisfied and that $\alpha+\beta=1$ with $\beta>0$. Then, we have the quadratic strong law
\begin{equation}\label{hsna_th_qsl_res1}
    \lim_{n\to +\infty} \dfrac{1}{\log(n)} \sum_{k=1}^n (\widehat{\theta}_k-\theta )(\widehat{\theta}_k-\theta )^T =\left(\nabla^2 G(\theta)\right)^{-1} \qquad \text{a.s.}
\end{equation}
%Consequently, 
%\begin{equation}\label{hsna_th_qsl_res2}
%    \lim_{n\to +\infty} \dfrac{1}{\log(n)} \sum_{k=1}^n \lVert \widehat{\theta}_k-\theta \rVert^2 =
%    \Tr \left[\left(\nabla^2 G(\theta)\right)^{-1}\right] \qquad \text{a.s.}
%\end{equation}
Moreover, we also have 
\begin{equation}\label{hsna_th_qsl_res2}
    \lim_{n\to +\infty} \dfrac{1}{\log(n)} \sum_{k=1}^n \Big(G(\widehat{\theta}_k)-G(\theta)\Big)  =
    \dfrac{d}{2} \qquad \text{a.s.}
\end{equation}
\end{theorem}
\begin{proof}
The proof of Theorem \ref{hsna_th_qsl} can be found in Appendix \ref{app_hsna_th_qsl}.
\end{proof}
Hazan \cite{hazan2007logarithmic} has shown similar logarithmic upper bound of the cumulative excess risk \eqref{hsna_th_qsl_res2} only for the deterministic version of the ONS algorithm under the exp-concavity assumption \cite{hazan2016introduction}. This condition is satisfied for the logistic regression problem as soon as the sequence $(\widehat{\theta}_n)$ belongs to a bounded convex set whose diameter must be known. Therefore, our result appears more precise by providing an exact limit of that cumulative excess risk.

\section{Simulations}\label{sec:experiments}
\vspace{-1cm} 
In this section, we investigate the behavior of several algorithms on different datasets for logistic regression. In the first subsection, we focus our attention on the synthetic dataset. The second one presents the results obtained on real dataset. The goal in both cases, is to compare the performances of five algorithms as baseline : the stochastic gradient descent (SGD), the stochastic Newton algorithm (SN), the truncated stochastic Newton algorithm (TSN), the online Newton step (ONS) and our hybrid stochastic Newton algorithm. 
\par
We highlight that in all situations, the TSN algorithm is running as recommended in \cite{bercu2020efficient}.
Moreover, we implement the SGD algorithm with the decreasing step-size $1/n$. We recall that the ONS algorithm is exactly our hybrid stochastic Newton algorithm in the special case $\alpha=0$ and $\beta=1$. In what follows, concerning our hybrid stochastic Newton algorithm, we place ourselves in the situation where $\alpha+\beta=1$ and we specify below the values used for the two parameters $\alpha$ and $\beta$ in each case.

\vspace{-1cm} 
\subsection{Synthetic data}
\vspace{-1cm} 

We illustrate the robustness of our hybrid stochastic Newton algorithm when the dimension $d$ increases with synthetic dataset. For that purpose, we test several dimensions $d$ belonging to $\{10,50,100,200\}$. The feature $\Phi$ is a random vector with $d$ independent coordinates identically and uniformly distributed on the interval $[0,1]$. Moreover, the true unknown parameter $\theta$ is chosen such that its components are independent with the same uniform discrete distribution on $[\![-10,10]\!]$. We consider this large range in order to have a model which is not well-conditioned. All algorithms have been initialized at the origin of $\R^d$. We summarize in Table \ref{synthetic_param} the values of $\alpha$ and $\beta$ for our hybrid stochastic Newton algorithm. 

\begin{table}[t]%[H]%
\caption{Values of $\alpha$ and $\beta$ used for each dimension}
\label{synthetic_param}
\begin{center}
\begin{tabular}{@{}ccccc@{}}
\hline
Dimension $d$  & 10 & 50 & 100 & 200 \\
\hline
$\alpha$    & $10^{-10}$ & $10^{-10}$ & 0.25  & 0.9  \\
$\beta$    & $1-10^{-10}$  & $1-10^{-10}$  & 0.75  & 0.1  \\
\hline
\end{tabular}
\end{center}
\end{table}

\noindent We evaluate the different algorithms with the mean squared error $\E[\lVert \widehat{\theta}_{n} -\theta \rVert^2]$ that we approximate by its empirical version on 100 samples. The results are displayed in Figure \ref{fig:all_mse_synthetic}.

\begin{figure}[H]%[t]
    \centering
    \begin{subfigure}[b]{0.49\textwidth}
        \includegraphics[width=\textwidth]{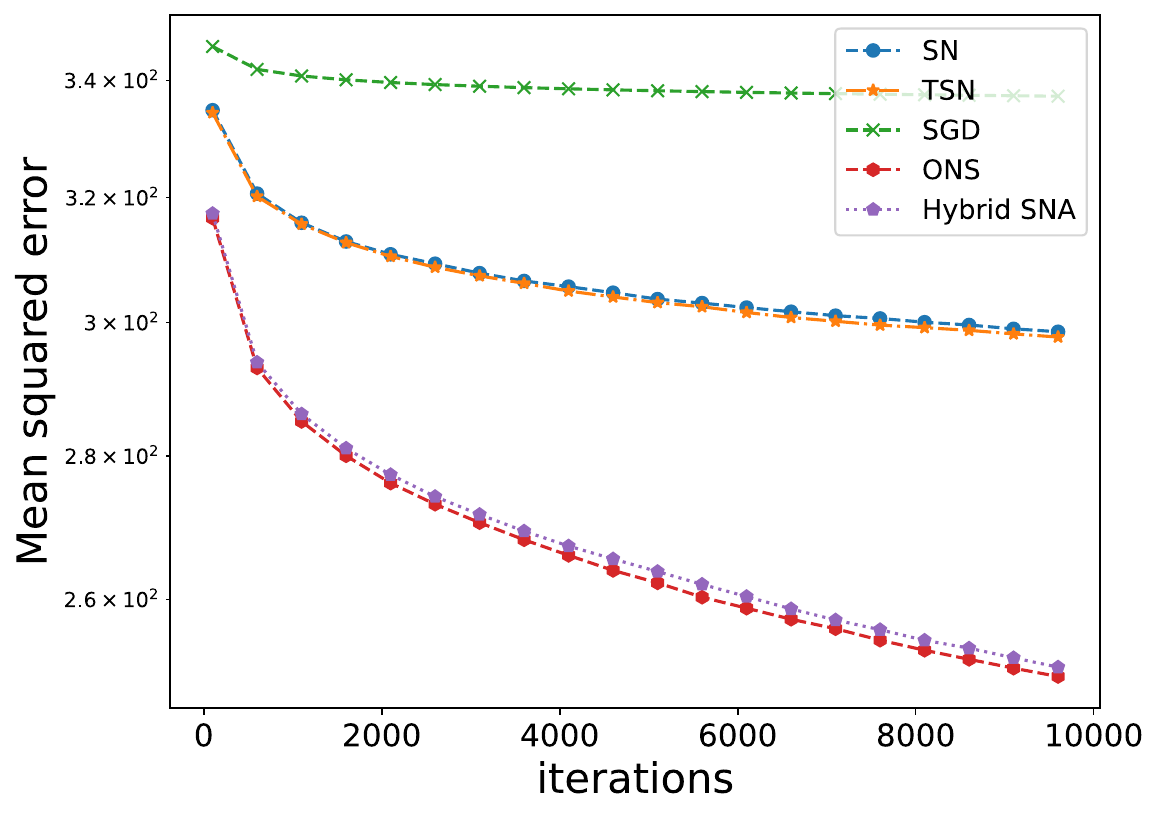}
        %\captionsetup{labelformat=empty}
        \caption{$d=10$}
    \end{subfigure}
    \hfill
    \begin{subfigure}[b]{0.49\textwidth}
        \includegraphics[width=\textwidth]{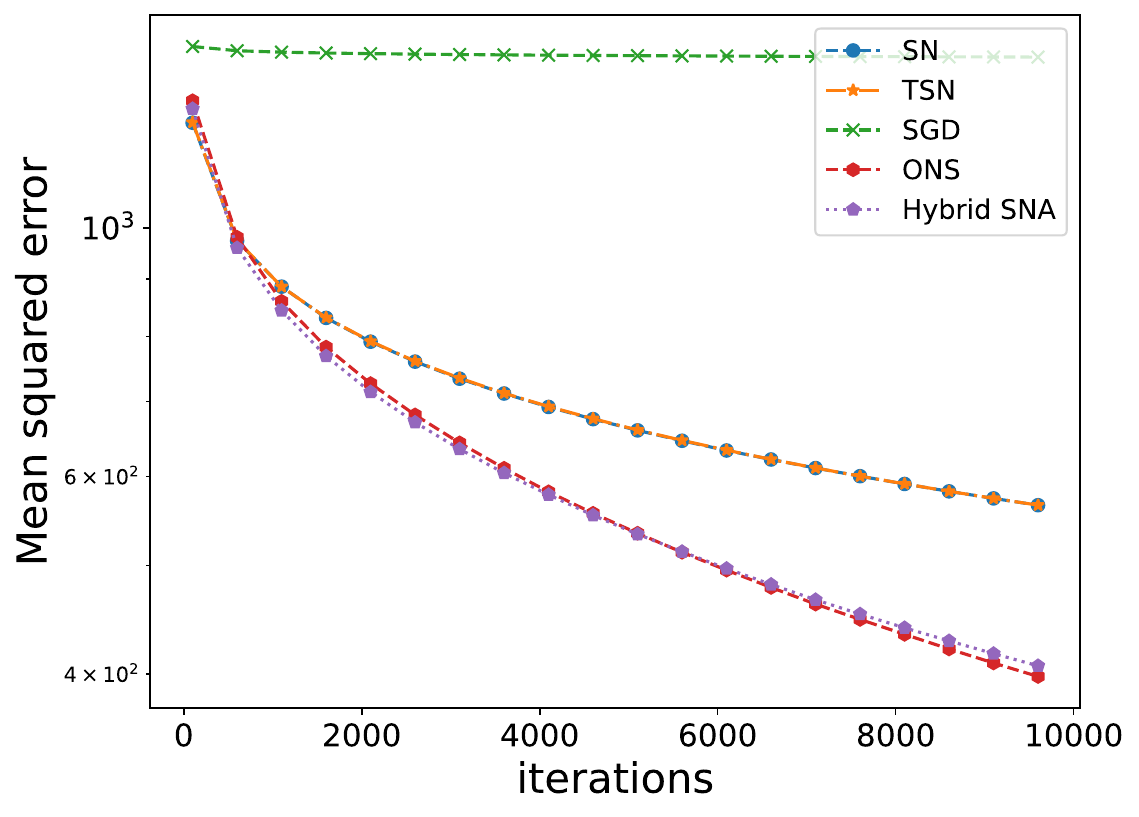}
        %\captionsetup{labelformat=empty}
        \caption{$d=50$}
    \end{subfigure}
    \hfill
    %\vspace{1cm}
    \\
    \begin{subfigure}[b]{0.49\textwidth}
        \includegraphics[width=\textwidth]{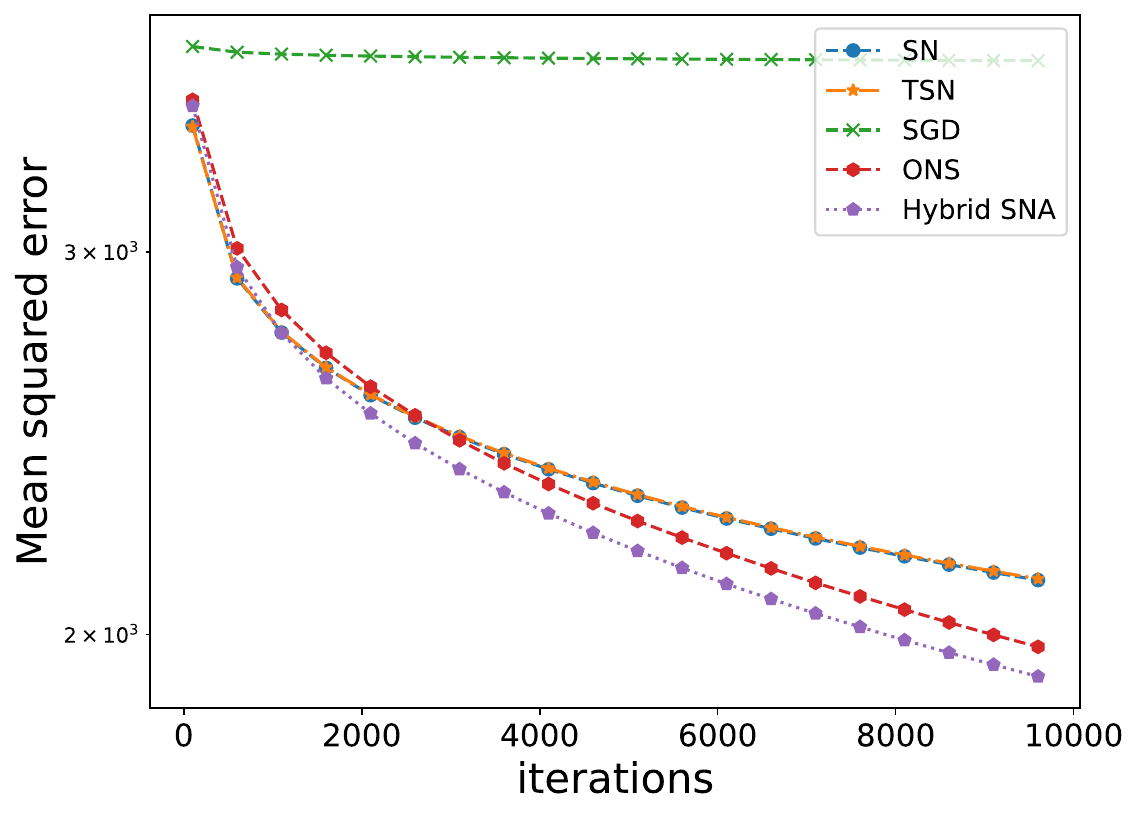}
        %\captionsetup{labelformat=empty}
        \caption{$d=100$}
    \end{subfigure}
    \hfill
    \begin{subfigure}[b]{0.49\textwidth}
        \includegraphics[width=\textwidth]{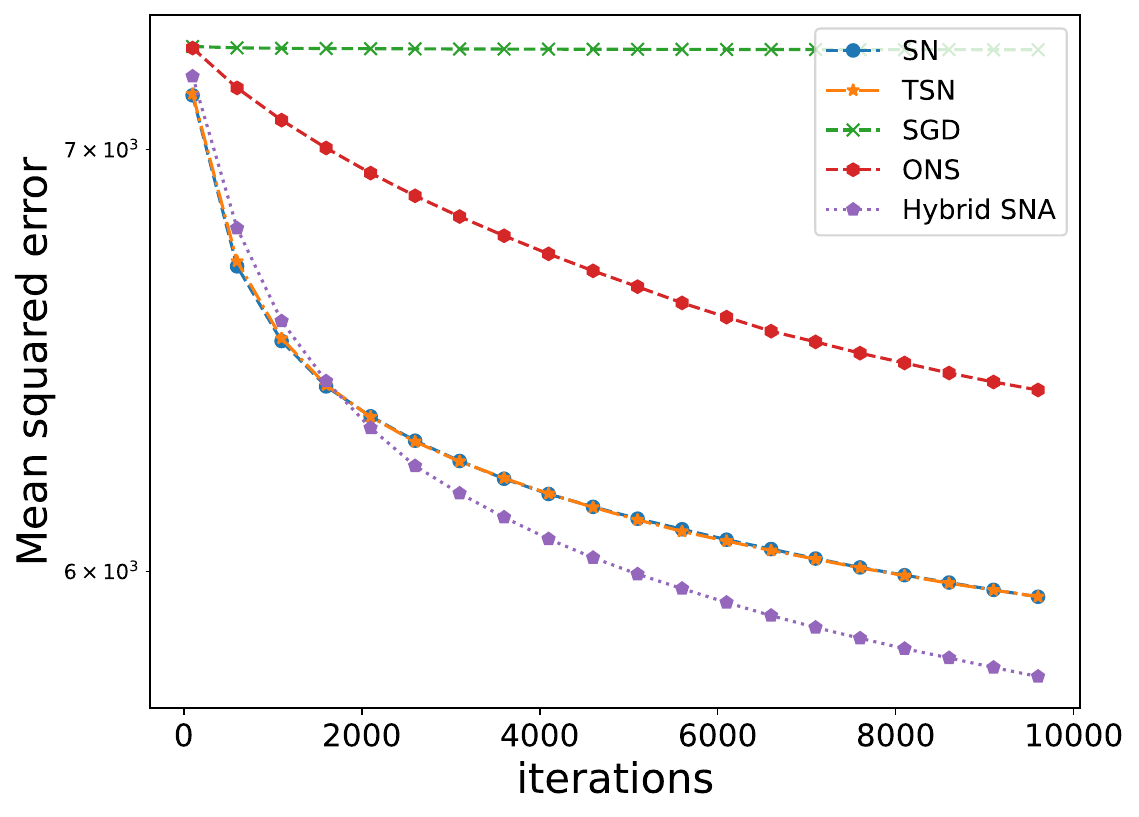}
        %\captionsetup{labelformat=empty}
        \caption{$d=200$}
    \end{subfigure}
    \hfill
    \caption{Evolution of the mean squared error with respect to the iteration.}
    \label{fig:all_mse_synthetic}
\end{figure}
As shown in Figure \ref{fig:all_mse_synthetic}, the general trend is that our hybrid stochastic Newton algorithm performs very well and is always at least close to the best method between the SN algorithm and the ONS. The reason is that we can choose the parameters $\alpha$ or $\beta$ close to 1 and in such a way that $\alpha+\beta=1$. This is the case for $d=10$ and $d=50$ where we observe that the ONS algorithm and our algorithm coincide and yield better performances than the stochastic Newton algorithm and its truncated version.  However, our hybrid stochastic Newton algorithm has the best decreasing of the mean squared error when the dimension $d$ is equal to $100$. Moreover, we observe that for $d=200$, our hybrid stochastic Newton algorithm with $\alpha=0.9$ and $\beta=0.1$ is significantly faster, then followed by the SN algorithm and the ONS. Hence, this interesting finding shows that our algorithm can outperform the standard stochastic Newton algorithm particularly in high-dimensional context.

\subsection{Real data}
We follow the same procedure on two real datasets. \textbf{Forest cover-type} \cite{blackard1999comparative} contains feature vectors of dimension $d = 54$ and consists of predicting forest cover type from cartographic variables only. There are 7 classes, but as in \cite{de2021stochastic}, we focus on classifying 2 versus all others. There are $n = 581012$ observations and we randomly split for training ($99 \%$) and testing ($1 \%$). \textbf{Adult income} \cite{kohavi1996scaling} is a binary classification dataset designed to the prediction task which is to determine whether a person makes an annual income smaller or bigger than 50K. All categorical variables are transformed into binary indicators which leads to $d=100$. We use the predefined split between training (32561 observations) and testing (16281 observations). For each dataset, we standardize $\Phi$ such that each feature ranges from 0 to 1. We run our hybrid stochastic Newton algorithm with  $\alpha=1-10^{-10}$ and $\beta=10^{-10}$ for both datasets.

The comparison of the algorithms is carried out using $\E[ G( \widehat{\theta}_{n})]- G(\theta)$, the expected excess risk estimated on the test set thanks to the standard Monte Carlo procedure with 100 samples. Figure \ref{fig:all_risk_real} shows the performances of the stochastic Newton algorithms. 

\begin{figure}[t]
    \centering
    \begin{subfigure}[b]{0.49\textwidth}
        \includegraphics[width=\textwidth]{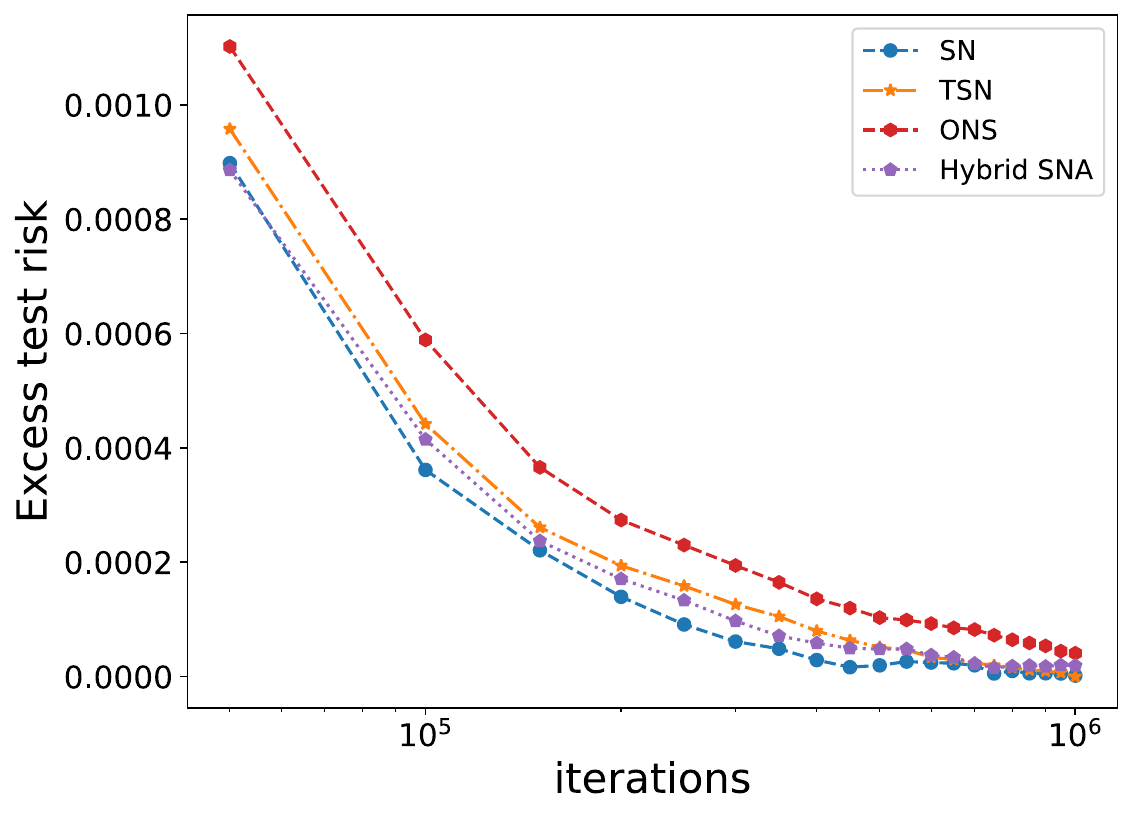}
        %\captionsetup{labelformat=empty}
        \caption{Forest cover-type}
    \end{subfigure}
    \hfill
    \begin{subfigure}[b]{0.49\textwidth}
        \includegraphics[width=\textwidth]{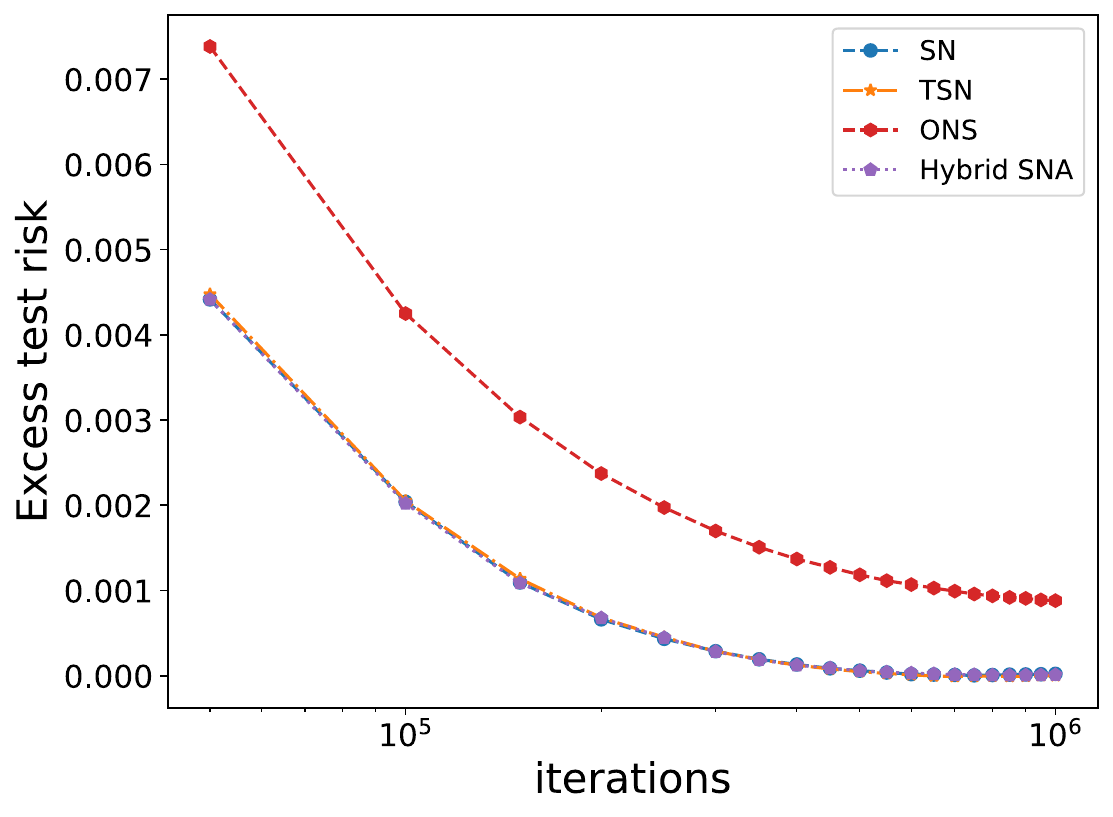}
        %\captionsetup{labelformat=empty}
        \caption{Adult income}
    \end{subfigure}
    \caption{Evolution of the expected excess risk evaluated on the test set with respect to the iteration. We observe that our hybrid stochastic Newton and the SN algorithms are very close in both real-world datasets.}
    \label{fig:all_risk_real}
\end{figure}

\section{Conclusion}
In this paper, we have investigated the theoretical and numerical properties of a new hybrid stochastic Newton algorithm for solving the logistic regression problem which is widely used and has many applications, including in the fields of machine learning, public health, social sciences, ecology, and econometrics. Furthermore, we have proposed a new approach on the analysis of second-order stochastic algorithms, which consists to update the Hessian matrix estimate before updating the algorithm. This approach combined with the modification made in the Hessian matrix estimation, allows us to establish several convergence guarantees while avoiding the useless introduction of truncation as it was done in the previous literature.
Therefore, we provide the almost sure convergence of our new algorithm and his asymptotic efficiency via a central limit theorem. We highlight that our hybrid stochastic Newton algorithm also leads to improvement of the almost sure rate of convergence to the Hessian matrix. Moreover, the numerical experiments performed on simulated and real-world datasets, provide additional empirical evidence of the relevance of our proposed method and validate its applicability. In particular, we find that our hybrid stochastic Newton algorithm achieves performance at least equivalent to that of the best method among the ONS and the standard stochastic Newton algorithm. Furthermore, our method remains adapted when the dimension of the feature vector increases. It would therefore be worthwhile to extend our approach to a broader class of optimization problems for machine learning.

\begin{appendix}

%%%%%%%%%%%%%%%%%%%%%%%%%%%%%%%%%%%%%%%%%%%%%%%%%%%%%%%%%%%%%%%%%%%%%%%%%

\section{Proof of Theorem \ref{hsna_th_convps}}\label{app_hsna_th_convps}

%%%%%%%%%%%%%%%%%%%%%%%%%%%%%%%%%%%%%%%%%%%%%%%%%%%%%%%%%%%%%%%%%%%%%%%%%

\begin{proof} Let us define, for all $n\geqslant 1$, the positive random variable
\begin{equation} \label{hsna_conv_ps_eq1}
    T_n= G(\widehat{\theta}_{n})-G(\theta).
\end{equation}
By a Taylor expansion of the twice continuously differentiable functional $G$, there exists $\xi_{n+1}\in \R^d$ such that
\begin{equation}\label{hsna_conv_ps_eq2}
    T_{n+1}= T_n +\nabla G(\widehat{\theta}_{n})^T(\widehat{\theta}_{n+1}-\widehat{\theta}_{n})+ \dfrac{1}{2} (\widehat{\theta}_{n+1}-\widehat{\theta}_{n})^T\nabla^2 G(\xi_{n+1})(\widehat{\theta}_{n+1}-\widehat{\theta}_{n}).
\end{equation}
However, we clearly have from \eqref{general_hessian_logistic}
\[\Tr[\nabla^2 G(\xi_{n+1})] \leqslant \dfrac{1}{4} \E[\lVert\Phi\rVert^2]. \]
Hence, we obtain that
\begin{equation}\label{hsna_conv_ps_eq3}
    T_{n+1}\leqslant T_n+\nabla G(\widehat{\theta}_{n})^T(\widehat{\theta}_{n+1}-\widehat{\theta}_{n})+\dfrac{1}{8}\E[\lVert\Phi\rVert^2]\lVert \widehat{\theta}_{n+1}-\widehat{\theta}_{n} \rVert^2.
\end{equation}
The recursive definition \eqref{hsna}, combined with the previous inequality, leads to 
\begin{equation}\label{hsna_conv_ps_eq4}
    T_{n+1}\leqslant T_n-\nabla G(\widehat{\theta}_{n})^T S_{n+1}^{-1} Z_{n+1}+\dfrac{1}{8}\E[\lVert\Phi\rVert^2](Y_{n+1}-\pi(\widehat{\theta}_n^T\Phi_{n+1}))^2\Phi_{n+1}^T S_{n+1}^{-2}\Phi_{n+1},
\end{equation}
where $Z_{n+1}=(\pi(\widehat{\theta}_n^T\Phi_{n+1})-Y_{n+1})\Phi_{n+1}$. By taking the conditional expectation on both sides of \eqref{hsna_conv_ps_eq4}, we obtain that almost surely
\begin{equation}\label{hsna_conv_ps_eq5}
    \E\Big[T_{n+1}\Big\lvert \F_n\Big] \leqslant T_n -\E\Big[\nabla G(\widehat{\theta}_{n})^T S_{n+1}^{-1} Z_{n+1}\Big\lvert \F_n\Big] +\dfrac{1}{8}\E[\lVert\Phi\rVert^2]\E\Big[b_{n+1}\Phi_{n+1}^T S_{n+1}^{-2}\Phi_{n+1} \Big\lvert \F_n\Big],
\end{equation}
where $\F_{n} = \sigma\left( (\Phi_1,Y_1),\dots,(\Phi_n,Y_n)\right)$.
The last two terms in this inequality require a particular attention because $S_{n+1}^{-1}$ is not $\F_n$-measurable. However, we obtain from \eqref{hsna_S_n_inv} the very useful relation
\begin{equation}\label{hsna_conv_ps_eq6}
    S_{n+1}^{-1}\Phi_{n+1}=\dfrac{1}{1+g_{n+1}}S_{n}^{-1}\Phi_{n+1},
\end{equation}
with $g_{n+1}=c_{n+1}\Phi_{n+1}^T S_{n}^{-1} \Phi_{n+1}$. Moreover, remark also that Assumption \ref{sna_cond2} implies
\begin{equation}\label{hsna_conv_ps_eq7}
    \dfrac{1}{1+g_{n+1}} \geqslant \dfrac{1}{1+(\alpha+\beta)d_\Phi^2 \valmax(S_{n}^{-1})},
\end{equation}
where $\valmax (S_{n}^{-1})$ denotes the largest eigenvalue of $S_{n}^{-1}$. Hereafter, one deduces from the definition \eqref{general_grad_logistic} that almost surely,
\begin{equation}\label{hsna_conv_ps_eq8}
    \E\Big[Z_{n+1} \Big \lvert \F_n\Big]= \nabla G(\widehat{\theta}_{n}) \qquad \text{a.s.}
\end{equation}
We now have everything we need to prove, through a case study, that
\begin{equation}\label{hsna_conv_ps_eq9}
    -\E\Big[\nabla G(\widehat{\theta}_{n})^T S_{n+1}^{-1} Z_{n+1}\Big\lvert \F_n\Big] \leqslant - \dfrac{1}{1+(\alpha+\beta)d_\Phi^2 \valmax(S_{n}^{-1})} \nabla G(\widehat{\theta}_{n})^T S_{n}^{-1}\nabla G(\widehat{\theta}_{n}) \quad \text{a.s.}
\end{equation}
On the one hand, assume that $\nabla G(\widehat{\theta}_{n})^T S_{n}^{-1} Z_{n+1} \geqslant 0$.
In this case, we have from \eqref{hsna_conv_ps_eq6} and \eqref{hsna_conv_ps_eq7} that
\begin{equation}\label{hsna_conv_ps_eq10}
    -\nabla G(\widehat{\theta}_{n})^T S_{n+1}^{-1} Z_{n+1}\leqslant -\dfrac{1}{1+(\alpha+\beta)d_\Phi^2  \valmax(S_{n}^{-1})} \nabla G(\widehat{\theta}_{n})^T S_{n}^{-1} Z_{n+1},
\end{equation}
which immediately gives \eqref{hsna_conv_ps_eq9} by using the equation \eqref{hsna_conv_ps_eq8}. On the other hand, suppose that $\nabla G(\widehat{\theta}_{n})^T S_{n}^{-1} Z_{n+1} < 0$.
Here, since $(1+g_{n+1})^{-1}\leqslant 1$,  we deduce once again from \eqref{hsna_conv_ps_eq6} that
\[\nabla G(\widehat{\theta}_{n})^T S_{n+1}^{-1}Z_{n+1} \geqslant \nabla G(\widehat{\theta}_{n})^T S_{n}^{-1} Z_{n+1},\]
Then, by taking the conditional expectation and using \eqref{hsna_conv_ps_eq8} on the right-hand side, we obtain that almost surely
\begin{equation}\label{hsna_conv_ps_eq11}
    -\E\Big[\nabla G(\widehat{\theta}_{n})^T S_{n+1}^{-1} Z_{n+1} \Big\lvert \F_n\Big] \leqslant - \nabla G(\widehat{\theta}_{n})^T S_{n}^{-1}\nabla G(\widehat{\theta}_{n}).
\end{equation}
Therefore, since $S_{n}^{-1}$ is a positive definite matrix, this also yields \eqref{hsna_conv_ps_eq9}. Consequently, the inequality \eqref{hsna_conv_ps_eq9} holds for all $n\geqslant 1$. Furthermore, we have from \eqref{hsna_conv_ps_eq6} that
\begin{equation}\label{hsna_conv_ps_eq12}
    \Phi_{n+1}^T S_{n+1}^{-2}\Phi_{n+1}=\dfrac{1}{(1+g_{n+1})^2}\Phi_{n+1}^TS_{n}^{-2}\Phi_{n+1}.
\end{equation}
Moreover, the elementary relation which immediately follows from the definition \eqref{hsna_c_n} is that for all $n\geqslant 1$,
\begin{equation}\label{hsna_conv_ps_eq13}
    b_{n+1} \leqslant \dfrac{1}{\beta}c_{n+1}.
\end{equation}
Hence, we get from \eqref{hsna_conv_ps_eq12} and \eqref{hsna_conv_ps_eq13} that almost surely
\begin{align}
    \E\Big[b_{n+1}\Phi_{n+1}^T S_{n+1}^{-2}\Phi_{n+1} \Big\lvert \F_n\Big]& \leqslant \dfrac{1}{\beta} \E\left[c_{n+1}\dfrac{1}{(1+g_{n+1})^2}\Phi_{n+1}^T S_{n}^{-2}\Phi_{n+1} \Big\lvert \F_n\right] \nonumber\\
    &\leqslant \dfrac{1}{\beta} \E\left[\dfrac{c_{n+1}}{1+g_{n+1}}\Phi_{n+1}^T S_{n}^{-2}\Phi_{n+1} \Big\lvert \F_n\right].\label{hsna_conv_ps_eq14}
\end{align}
We deduce from the recursive relation \eqref{hsna_S_n} that all eigenvalues of $S_n^{-1}$ are smallest than 1.
Therefore, by putting together the inequalities \eqref{hsna_conv_ps_eq5}, \eqref{hsna_conv_ps_eq9} and \eqref{hsna_conv_ps_eq14}, we obtain that for all $n\geqslant 1$
 \begin{equation}\label{hsna_conv_ps_eq15}
    \E\Big[T_{n+1} \Big \lvert \F_n\Big]\leqslant T_n-\mathcal{B}_n+\mathcal{A}_n \qquad \text{a.s.},
\end{equation}
where
\begin{equation}\label{hsna_conv_ps_eq16}
    \mathcal{B}_n= \dfrac{1}{1+(\alpha+\beta)d_\Phi^2} \nabla G(\widehat{\theta}_{n})^T S_{n}^{-1}\nabla G(\widehat{\theta}_{n}),
\end{equation}
and 
\begin{equation}\label{hsna_conv_ps_eq17}
    \mathcal{A}_n= \dfrac{d_\Phi^2}{8\beta} \E\left[\dfrac{c_{n+1}}{1+g_{n+1}}\Phi_{n+1}^T S_{n}^{-2}\Phi_{n+1} \Big\lvert \F_n\right].
\end{equation}
Furthermore, it is easy to observe from \eqref{hsna_S_n_inv} that
\begin{equation}\label{hsna_conv_ps_eq18}
    \Tr[S_{n+1}^{-1}]=\Tr[S_{n}^{-1}]-\dfrac{c_{n+1}}{1+g_{n+1}} \Phi_{n+1}^T S_{n}^{-2} \Phi_{n+1},
\end{equation}
which implies that
\begin{equation}\label{hsna_conv_ps_eq19}
    \E\Big[\Tr[S_{n+1}^{-1}] \Big\lvert \F_n \Big]=\Tr[S_{n}^{-1}]-\E\left[\dfrac{c_{n+1}}{1+g_{n+1}}\Phi_{n+1}^T S_{n}^{-2}\Phi_{n+1} \Big\lvert \F_n\right].
\end{equation}
Moreover, denote by $W_n$ the positive random variable such that
\begin{equation}
    W_n=\Tr[S_{n}^{-1}]+ \sum_{k=1}^{n} \E\left[\dfrac{c_{k+1}}{1+g_{k+1}}\Phi_{k+1}^T S_{k}^{-2}\Phi_{k+1} \Big\lvert \F_k\right].
\end{equation}
We clearly have from \eqref{hsna_conv_ps_eq19} that $\E[W_{n+1}\lvert \F_n]=W_n$ which means that $(W_n)$ is a positive martingale sequence. Hence, we deduce that the sequence $(W_n)$ converges almost surely towards an integrable random variable. In particular, this immediately leads to,
\begin{equation}\label{hsna_conv_ps_eq20}
    \sum_{n=1}^\infty \mathcal{A}_n < +\infty \qquad a.s.
\end{equation}
The three sequences $(T_n)$, $(\mathcal{A}_n)$ and $(\mathcal{B}_n)$ are positive sequences of random variables adapted to $(\F_n)$. Therefore, we conclude from the Robbins-Siegmund Theorem \citep{robbins1971convergence} on inequality \eqref{hsna_conv_ps_eq15} that almost surely, $(T_n)$ converges towards a finite random variable and 
\begin{equation}\label{hsna_conv_ps_eq21}
    \sum_{n=1}^\infty \mathcal{B}_n < +\infty \qquad \text{a.s.}
\end{equation}
Furthermore, we clearly have from \eqref{hsna_conv_ps_eq16} that
\begin{equation*}
    \mathcal{B}_n\geqslant \dfrac{\big(\valmax(S_{n})\big)^{-1}\lVert \nabla G(\widehat{\theta}_{n})\rVert^2}{1+(\alpha+\beta)d_\Phi^2} .
\end{equation*}
Then, it follows with \eqref{hsna_conv_ps_eq21} that
\begin{equation}\label{hsna_conv_ps_eq22}
    \mathlarger{\mathlarger{\sum}}_{n=1}^\infty \dfrac{\big(\valmax(S_{n})\big)^{-1}\lVert \nabla G(\widehat{\theta}_{n})\rVert^2}{1+(\alpha+\beta)d_\Phi^2} < +\infty \qquad \text{a.s.}
\end{equation}
Moreover, the definition \eqref{hsna_S_n} together with the strong law of large numbers associated with the sequence $(\Phi_n)$ imply that for $n$ large enough, there exists a positive constant $\lambda$ such that 
\begin{equation}\label{hsna_conv_ps_eq23}
    \valmax(S_{n}) \leqslant (\alpha+\beta)\lambda n,
\end{equation}
which leads to
\begin{equation}\label{hsna_conv_ps_eq24}
        \sum_{n=1}^\infty \big(\valmax(S_{n})\big)^{-1} = +\infty \qquad \text{a.s.}
\end{equation}
Hence, the combination of \eqref{hsna_conv_ps_eq22} and \eqref{hsna_conv_ps_eq24}
immediately ensures that $\nabla G(\widehat{\theta}_{n})$ converges towards 0. Then, we deduce that $\widehat{\theta}_{n}$ converges almost surely to the unique zero $\theta$ of the gradient, which completes the proof of \eqref{hsna_th_convps_res1}. Furthermore, we can easily rewrite $\overline{H}_n$ as follows
\begin{equation}\label{hsna_conv_ps_eq25}
    \overline{H}_n= \dfrac{1}{n} \sum_{k=1}^n a_k(\theta) \Phi_k \Phi_k^T+\dfrac{1}{n}\sum_{k=1}^n (a_k-a_k(\theta)) \Phi_k \Phi_k^T,
\end{equation}
where $a_k(\theta)=\pi(\theta^T\Phi_k)(1-\pi(\theta^T\Phi_k))$. Then, we immediately deduce from the strong law of large numbers associated with the sequence $(\Phi_n)$ that
\begin{equation}\label{hsna_conv_ps_eq26}
    \lim_{n\to +\infty} \dfrac{1}{n} \sum_{k=1}^n a_k(\theta) \Phi_k \Phi_k^T = \nabla^2 G(\theta) \qquad \text{a.s.}
\end{equation}
Moreover, we have that
\begin{equation}\label{hsna_conv_ps_eq27}
    \left \lVert \dfrac{1}{n}\sum_{k=1}^n (a_k-a_k(\theta)) \Phi_k \Phi_k^T \right\rVert \leqslant \dfrac{1}{n}\sum_{k=1}^n \lvert a_k-a_k(\theta)\rvert \lVert\Phi_k \rVert^2.
\end{equation}
However, we obtain from Lemma 6.2 of \cite{bercu2020efficient} that for all $n\geqslant 1$
\begin{equation}\label{hsna_conv_ps_eq27b}
    \lvert a_n-a_n(\theta)\rvert \leqslant \dfrac{1}{12\sqrt{3}} \lVert \Phi_n\rVert \lVert \widehat{\theta}_{n-1}-\theta\rVert \leqslant \dfrac{d_\Phi}{12\sqrt{3}} \lVert \widehat{\theta}_{n-1}-\theta\rVert.
\end{equation}
It implies with \eqref{hsna_conv_ps_eq27} that
\begin{equation*}
    \left \lVert \dfrac{1}{n}\sum_{k=1}^n (a_k-a_k(\theta)) \Phi_k \Phi_k^T \right\rVert
    \leqslant \dfrac{d_\Phi^3}{12\sqrt{3}}\dfrac{1}{n}\sum_{k=1}^n \lVert \widehat{\theta}_{k-1}-\theta\rVert.
\end{equation*}
Thanks to Toeplitz's lemma and the almost sure convergence of $\widehat{\theta}_{n}$ towards $\theta$, we obtain that
\begin{equation*}
    \lim_{n\to +\infty} \dfrac{1}{n}\sum_{k=1}^n \lVert \widehat{\theta}_{k-1}-\theta\rVert=0 \qquad \text{a.s.},
\end{equation*}
which leads to 
\begin{equation}\label{hsna_conv_ps_eq28}
    \lim_{n\to +\infty} \dfrac{1}{n}\sum_{k=1}^n (a_k-a_k(\theta)) \Phi_k \Phi_k^T =0 \qquad \text{a.s.}
\end{equation}
Consequently, we deduce from \eqref{hsna_conv_ps_eq26} and \eqref{hsna_conv_ps_eq28} that
\begin{equation*}
    \lim_{n\to +\infty} \overline{H}_n =\nabla^2 G(\theta) \qquad \text{a.s.}
\end{equation*}
Furthermore, we have that
\begin{align*}
    \overline{\Sigma}_n&=\dfrac{1}{n} \sum_{k=1}^n \Big(\pi(\widehat{\theta}_{k-1}^T\Phi_k)-Y_{k}\Big)^2 \Phi_k \Phi_k^T,\\
    &=\dfrac{1}{n} \sum_{k=1}^n \Big\{(\pi(\theta^T\Phi_k)-Y_{k})+(\pi(\widehat{\theta}_{k-1}^T\Phi_k)-\pi(\theta^T\Phi_k))\Big\}^2 \Phi_k \Phi_k^T,\\
    &=\dfrac{1}{n} \sum_{k=1}^n \Big(\pi(\theta^T\Phi_k)-Y_{k}\Big)^2 \Phi_k \Phi_k^T+\dfrac{1}{n} \sum_{k=1}^n \Big(\pi(\widehat{\theta}_{k-1}^T\Phi_k)-\pi(\theta^T\Phi_k)\Big)^2 \Phi_k \Phi_k^T \\
    &\,+\dfrac{2}{n} \sum_{k=1}^n             \Big(\pi(\theta^T\Phi_k)-Y_{k}\Big)\Big(\pi(\widehat{\theta}_{k-1}^T\Phi_k)-\pi(\theta^T\Phi_k)\Big) \Phi_k \Phi_k^T.
\end{align*}
In the same way as in \eqref{hsna_conv_ps_eq26}, we obtain that
\begin{equation}\label{hsna_conv_ps_eq29}
    \lim_{n\to +\infty} \dfrac{1}{n} \sum_{k=1}^n \Big(\pi(\theta^T\Phi_k)-Y_{k}\Big)^2 \Phi_k \Phi_k^T = \Sigma(\theta) \qquad \text{a.s.},
\end{equation}
where $\Sigma(\theta)=\E[(\pi(\theta^T \Phi)-Y)^2\Phi\Phi^T]$.
Then, it is easy to prove by using the mean value theorem that for all $n\geqslant 1$
\begin{equation}\label{hsna_conv_ps_eq29b}
    \lvert \pi(\widehat{\theta}_{n-1}^T\Phi_n)-\pi(\theta^T\Phi_n) \rvert \leqslant \dfrac{1}{4} \big\lvert (\widehat{\theta}_{n-1}-\theta)^T\Phi_n\big\rvert \leqslant \dfrac{d_\Phi}{4} \lVert \widehat{\theta}_{n-1}-\theta\rVert. 
\end{equation}
Hence, it follows that
\begin{equation*}
    \left\lVert \dfrac{1}{n} \sum_{k=1}^n (\pi(\widehat{\theta}_{k-1}^T\Phi_k)-\pi(\theta^T\Phi_k))^2 \Phi_k \Phi_k^T \right\rVert 
    \leqslant \dfrac{d_\Phi^4}{16}\dfrac{1}{n} \sum_{k=1}^n \lVert \widehat{\theta}_{k-1}-\theta\rVert^2.
\end{equation*}
As before, it ensures that
\begin{equation}\label{hsna_conv_ps_eq30}
    \lim_{n\to +\infty} \dfrac{1}{n} \sum_{k=1}^n \Big(\pi(\widehat{\theta}_{k-1}^T\Phi_k)-\pi(\theta^T\Phi_k)\Big)^2 \Phi_k \Phi_k^T = 0 \qquad  \text{a.s.}
\end{equation}
Moreover, we deduce from \eqref{hsna_conv_ps_eq29} and \eqref{hsna_conv_ps_eq30} that
\begin{equation}\label{hsna_conv_ps_eq31}
    \lim_{n\to +\infty} \dfrac{1}{n}\sum_{k=1}^n \Big(\pi(\theta^T\Phi_k)-Y_{k}\Big)\Big(\pi(\widehat{\theta}_{k-1}^T\Phi_k)-\pi(\theta^T\Phi_k)\Big) \Phi_k \Phi_k^T= 0 \qquad  \text{a.s.}
\end{equation}
Finally, by putting together \eqref{hsna_conv_ps_eq29}, \eqref{hsna_conv_ps_eq30} and \eqref{hsna_conv_ps_eq31}, we obtain that
\begin{equation*}
    \lim_{n\to +\infty} \overline{\Sigma}_n=\Sigma (\theta)=\nabla^2 G(\theta) \qquad  \text{a.s.}
\end{equation*}
For the last convergence result of Theorem \ref{hsna_th_convps}, it is enough to remark that
\begin{equation*}
    \overline{S}_n=\alpha\overline{H}_n+\beta\overline{\Sigma}_n+\dfrac{1}{n}\I_d,
\end{equation*}
and conclude by using \eqref{hsna_th_convps_res2}.
\end{proof}

%%%%%%%%%%%%%%%%%%%%%%%%%%%%%%%%%%%%%%%%%%%%%%%%%%%%%%%%%%%%%%%%%%%%%%%%%

\section{Proof of Theorem \ref{hsna_th_rates}}\label{app_hsna_th_rates}

%%%%%%%%%%%%%%%%%%%%%%%%%%%%%%%%%%%%%%%%%%%%%%%%%%%%%%%%%%%%%%%%%%%%%%%%%

\begin{proof}
We obtain from \eqref{hsna} that for all $n\geqslant 0$,
\begin{align*}
     \widehat{\theta}_{n+1}-\theta &=\widehat{\theta}_n-\theta -\dfrac{1}{n+1}\overline{S}_{n+1}^{-1}Z_{n+1} \\
     &=\widehat{\theta}_n-\theta -\dfrac{1}{n+1}D_{n+1}-\dfrac{1}{n+1}S^{-1}Z_{n+1},
\end{align*}
where $S=\nabla^2 G(\theta)$ and
\begin{equation*}
    D_{n+1}=\left(\overline{S}_{n+1}^{-1}-S^{-1} \right)Z_{n+1}.
\end{equation*}
Consequently, we obtain that for all $n\geqslant 0$,
\begin{equation}\label{hsna_th_rates_eq2}
    \widehat{\theta}_{n+1}-\theta =\widehat{\theta}_n-\theta -\dfrac{1}{n+1}D_{n+1}-\dfrac{1}{n+1} S^{-1}\Big(\nabla G(\widehat{\theta}_{n})+\varepsilon_{n+1}\Big),
\end{equation}
where
\begin{equation}\label{hsna_th_rates_eq3}
    \varepsilon_{n+1}=Z_{n+1}-\nabla G(\widehat{\theta}_{n}).
\end{equation}
We immediately have from \eqref{hsna_conv_ps_eq8} that $\E[\varepsilon_{n+1}|\F_n]=0$ which means that $(\varepsilon_n)$ is a martingale difference sequence adapted to $(\F_n)$.
Moreover, denote by $\delta_n$ the remainder of the Taylor’s expansion of the gradient
\begin{equation}\label{hsna_th_rates_eq4}
    \delta_n=\nabla G(\widehat{\theta}_{n})-\nabla^2 G(\theta)(\widehat{\theta}_n-\theta)=\nabla G(\widehat{\theta}_{n})-S(\widehat{\theta}_n-\theta).
\end{equation}
It follows from \eqref{hsna_th_rates_eq2} and \eqref{hsna_th_rates_eq4} that
\begin{equation}
\widehat{\theta}_{n+1}-\theta =\left(1-\dfrac{1}{n+1}\right)(\widehat{\theta}_n-\theta) -\dfrac{1}{n+1}D_{n+1}-\dfrac{1}{n+1} S^{-1}\Big(\delta_n+\varepsilon_{n+1}\Big).
\end{equation}
which implies that
\begin{equation}\label{hsna_th_rates_eq5a}
\widehat{\theta}_{n+1}-\theta = -\dfrac{1}{n+1} \sum_{k=0}^n D_{k+1}
-\dfrac{1}{n+1}S^{-1} \sum_{k=0}^n \Big(\delta_k+\varepsilon_{k+1}\Big),
\end{equation}
In order to bring out a martingale term in \eqref{hsna_th_rates_eq5a}, it is necessary to create a shift in this equation by splitting $D_{n+1}$ into two terms 
\begin{equation}
    D_{n+1}= \Big(\overline{S}_{n+1}^{-1}-\overline{S}_{n}^{-1}\Big)Z_{n+1}+\Big(\overline{S}_{n}^{-1}-S^{-1}\Big)Z_{n+1}.
\end{equation}
Then, we obtain from \eqref{hsna_th_rates_eq3} and \eqref{hsna_th_rates_eq5a} that for all $n\geqslant 0$,
\begin{equation}\label{hsna_th_rates_eq5}
  \widehat{\theta}_{n+1}-\theta=-\dfrac{1}{n+1}M_{n+1}-\Delta_n,
\end{equation}
where the martingale term
\begin{equation}\label{hsna_th_rates_eq6}
    M_{n+1}= \sum_{k=0}^n\overline{S}_{k}^{-1}\varepsilon_{k+1},
\end{equation}
and
\begin{equation*}\label{hsna_th_rates_eq7}
    \Delta_n=\dfrac{1}{n+1} \sum_{k=0}^n \Big(\overline{S}_{k+1}^{-1}-\overline{S}_{k}^{-1}\Big)Z_{k+1}+
    \dfrac{1}{n+1}\sum_{k=0}^n\Big(\overline{S}_{k}^{-1}-S^{-1}\Big)\nabla G(\widehat{\theta}_{k})
+\dfrac{1}{n+1}S^{-1} \sum_{k=0}^n \delta_k,
\end{equation*}
which can be rewritten as
\begin{equation}\label{hsna_th_rates_eq7b}
    \Delta_n=\dfrac{1}{n+1} \sum_{k=0}^n \Big(\overline{S}_{k+1}^{-1}-\overline{S}_{k}^{-1}\Big)Z_{k+1}+
    \dfrac{1}{n+1}\sum_{k=0}^n\Big(\overline{S}_{k}^{-1}-S^{-1}\Big)S(\widehat{\theta}_k-\theta)
+\dfrac{1}{n+1}\sum_{k=0}^n \overline{S}_{k}^{-1}\delta_k.
\end{equation}
The sequence $(M_n)$ is a locally square-integrable multidimensional martingale with predictable quadratic variation given by
\begin{equation}\label{hsna_th_rates_eq8b1}
    \langle M \rangle_n = \sum_{k=1}^n \overline{S}_{k}^{-1} \E\Big[\varepsilon_{k+1}\varepsilon_{k+1}^T \Big\lvert \F_k \Big] \overline{S}_{k}^{-1}.
\end{equation}
Furthermore, we have from the definition \eqref{hsna_th_rates_eq3} that almost surely
\begin{align}
    &\E\Big[\varepsilon_{n+1}\varepsilon_{n+1}^T \Big\lvert \F_n \Big]\nonumber\\
    &= \E\Big[Z_{n+1}Z_{n+1}^T \Big\lvert \F_n \Big]- \nabla G(\widehat{\theta}_{n})\left(\nabla G(\widehat{\theta}_{n})\right)^T,\nonumber\\
    &= \E\Big[(\pi(\widehat{\theta}_n^T\Phi_{n+1})-Y_{n+1})^2\Phi_{n+1}\Phi_{n+1}^T \Big\lvert \F_n \Big]- \nabla G(\widehat{\theta}_{n})\left(\nabla G(\widehat{\theta}_{n})\right)^T,\nonumber\\
    &= \E\Big[\Big((\pi(\widehat{\theta}_n^T\Phi_{n+1})-\pi(\theta^T\Phi_{n+1}))+ (\pi(\theta^T\Phi_{n+1})-Y_{n+1})\Big)^2\Phi_{n+1}\Phi_{n+1}^T\Big\lvert \F_n \Big]\nonumber \\
    & \;- \nabla G(\widehat{\theta}_{n})\left(\nabla G(\widehat{\theta}_{n})\right)^T,\nonumber\\
     &= \E\Big[\Big(\pi(\widehat{\theta}_n^T\Phi_{n+1})-\pi(\theta^T\Phi_{n+1})\Big)^2\Phi_{n+1}\Phi_{n+1}^T\Big\lvert \F_n \Big]+\nabla^2 G(\theta)- \nabla G(\widehat{\theta}_{n})\left(\nabla G(\widehat{\theta}_{n})\right)^T,\label{hsna_th_rates_eq8b2}
\end{align} 
since we already saw that
\begin{equation*}
    \E\Big[(\pi(\theta^T\Phi_{n+1})-Y_{n+1})^2\Phi_{n+1}\Phi_{n+1}^T\Big\lvert \F_n \Big]= \nabla^2 G(\theta) \qquad \text{a.s.},
\end{equation*}
and 
\begin{equation*}
    \E\Big[(\pi(\widehat{\theta}_n^T\Phi_{n+1})-\pi(\theta^T\Phi_{n+1}))(\pi(\theta^T\Phi_{n+1})-Y_{n+1})\Phi_{n+1}\Phi_{n+1}^T\Big\lvert \F_n \Big]= 0 \qquad \text{a.s.}
\end{equation*}
However, the almost sure convergence \eqref{hsna_th_convps_res1} implies by continuity that
\begin{equation}\label{hsna_th_rates_eq8b3}
    \lim_{n\to +\infty} \nabla G(\widehat{\theta}_{n})\left(\nabla G(\widehat{\theta}_{n})\right)^T=0 \qquad a.s.
\end{equation}
Moreover, we obtain via Assumption \ref{sna_cond2} and the inequality \eqref{hsna_conv_ps_eq29b} that
\begin{align}
    \left\lVert\E\Big[(\pi(\widehat{\theta}_n^T\Phi_{n+1})-\pi(\theta^T\Phi_{n+1}))^2\Phi_{n+1}\Phi_{n+1}^T\Big\lvert \F_n \Big] \right\rVert &\leqslant \dfrac{1}{16} \E\Big[\lVert \widehat{\theta}_{n}-\theta\rVert^2 \lVert\Phi_{n+1} \rVert^4 \Big\lvert \F_n \Big], \nonumber\\
    &\leqslant \dfrac{d_{\Phi}^4}{16} \lVert \widehat{\theta}_{n}-\theta\rVert^2 \label{hsna_th_rates_eq8b3b} \qquad \text{a.s.}
\end{align}
which ensures that
\begin{equation}\label{hsna_th_rates_eq8b4}
    \lim_{n\to +\infty} \E\Big[(\pi(\widehat{\theta}_n^T\Phi_{n+1})-\pi(\theta^T\Phi_{n+1}))^2\Phi_{n+1}\Phi_{n+1}^T\Big\lvert \F_n \Big]=0 \qquad \text{a.s.}
\end{equation}
Hence, we deduce from \eqref{hsna_th_rates_eq8b2}, \eqref{hsna_th_rates_eq8b3} and \eqref{hsna_th_rates_eq8b4} that
\begin{equation}\label{hsna_th_rates_eq8b5}
    \lim_{n\to +\infty}  \E\Big[\varepsilon_{n+1}\varepsilon_{n+1}^T \Big\lvert \F_n \Big] = \nabla^2 G(\theta) \qquad \text{a.s.}
\end{equation}
The almost sure convergence \eqref{hsna_th_convps_res3} to
\begin{equation}\label{hsna_th_rates_eq8b6}
        \lim_{n\to +\infty} \overline{S}_n^{-1}= (\alpha+\beta)^{-1}\left(\nabla^2 G(\theta)\right)^{-1} \qquad \text{a.s.}
\end{equation}
Consequently, we obtain from \eqref{hsna_th_rates_eq8b1}, \eqref{hsna_th_rates_eq8b5} and \eqref{hsna_th_rates_eq8b6} that
\begin{equation}\label{hsna_th_rates_eq8b7}
        \lim_{n\to +\infty} \dfrac{1}{n}\langle M\rangle_n= (\alpha+\beta)^{-2}\left(\nabla^2 G(\theta)\right)^{-1} \qquad \text{a.s.}
\end{equation}
Since the random vector $\Phi$ is bounded, it follows from the strong law of large numbers for multidimensional martingales given, e.g., by Theorem 4.3.16 in \cite{duflo1996algorithmes} that
\begin{equation}\label{hsna_th_rates_eq8b8}
    \lVert M_n \rVert^2=\mathcal{O}(n\log(n)) \qquad \text{a.s.}
\end{equation}
which means that there exists a finite positive random variable $E$ such that for all $n\geqslant 1$,
\begin{equation}\label{hsna_th_rates_eq8b9}
    \lVert M_{n} \rVert \leqslant E\sqrt{n\log(n)} \qquad \text{a.s.}
\end{equation}
Moreover, it follows from the Taylor expansion with integral remainder that
\begin{equation}\label{hsna_th_rates_eq8}
    \lVert \delta_n \rVert =o\Big(\lVert \widehat{\theta}_n-\theta \rVert \Big) \qquad \text{a.s.}
\end{equation}
Hereafter, denote for all $n\geqslant 1$,
\begin{equation}\label{hsna_th_rates_eq9}
    R_n= \dfrac{1}{n} \sum_{k=1}^n \Big(\overline{S}_{k+1}^{-1}-\overline{S}_{k}^{-1}\Big)Z_{k+1}.
\end{equation}
We have from \eqref{hsna_conv_ps_eq6} that
\begin{equation}\label{hsna_th_rates_eq9b}
    \overline{S}_{n+1}^{-1}\Phi_{n+1}=\dfrac{n+1}{(1+g_{n+1})n}\overline{S}_{n}^{-1}\Phi_{n+1},
\end{equation}
where $g_{n+1}=c_{n+1}\Phi_{n+1}^T S_{n}^{-1} \Phi_{n+1}$. Hence, one immediately obtains that
\begin{equation}\label{hsna_th_rates_eq9c}
    \Big(\overline{S}_{n+1}^{-1}-\overline{S}_{n}^{-1}\Big)\Phi_{n+1}=\dfrac{1}{1+g_{n+1}}\left(\dfrac{1}{n}-g_{n+1}\right)\overline{S}_{n}^{-1}\Phi_{n+1}.
\end{equation}
Consequently, we deduce from the fact that $Z_{n+1}=(\pi(\widehat{\theta}_n^T\Phi_{n+1})-Y_{n+1})\Phi_{n+1}$ together with \eqref{hsna_th_rates_eq9} and \eqref{hsna_th_rates_eq9c} that
\begin{equation}\label{hsna_th_rates_eq10}
    \lVert R_n \rVert \leqslant  \dfrac{1}{n}\sum_{k=1}^n \left \lvert \dfrac{1}{k}-g_{k+1}\right\rvert \cdot \lVert \overline{S}_{k}^{-1}\Phi_{k+1} \rVert.
\end{equation}
We obtain from Assumption \ref{sna_cond2} and the almost sure convergence of $\overline{S}_{n}$ given by \eqref{hsna_th_convps_res3} that there exists a positive constant $C$ such that for $n$ large enough
\begin{equation}\label{hsna_th_rates_eq10b}
    \lVert \overline{S}_{n}^{-1}\Phi_{n+1} \rVert \leqslant C  \qquad \text{a.s.}
\end{equation}
Moreover, we also obtain from \eqref{hsna_th_convps_res3} that there exists a positive constant $C$ such that for $n$ large enough
\begin{equation}\label{hsna_th_rates_eq10c}
    g_n\leqslant \dfrac{C}{n}.
\end{equation}
Then, it follows from \eqref{hsna_th_rates_eq10}, \eqref{hsna_th_rates_eq10b} and \eqref{hsna_th_rates_eq10c} that there exists a finite positive random variable $D$ such that
\begin{equation}\label{hsna_th_rates_eq11}
    \lVert R_n \rVert \leqslant \dfrac{D}{n}+\dfrac{C\log(n)}{n} \qquad \text{a.s.}
\end{equation}
Therefore, we deduce from \eqref{hsna_th_rates_eq7}, \eqref{hsna_th_rates_eq8} and \eqref{hsna_th_rates_eq11} that there exist a constant $0 < c < 1/2$ and a finite positive random variable $D$ such that for all $n \geqslant 1$,
\begin{equation}\label{hsna_th_rates_eq12}
    \lVert \Delta_n \rVert \leqslant \dfrac{D}{n}+\dfrac{C\log(n)}{n}+c L_n \qquad \text{a.s.},
\end{equation}
where
\begin{equation}\label{hsna_th_rates_eq13}
    L_n= \dfrac{1}{n} \sum_{k=1}^n \lVert \widehat{\theta}_k-\theta \rVert 
\end{equation}
Furthermore, we obtain from \eqref{hsna_th_rates_eq5}, \eqref{hsna_th_rates_eq8b9} and \eqref{hsna_th_rates_eq12} that almost surely
\begin{align}
    L_{n+1}&=\left( 1-\dfrac{1}{n+1}\right)L_n + \dfrac{1}{n+1}\lVert \widehat{\theta}_{n+1}-\theta \rVert,\nonumber\\
     &\leqslant \left( 1-\dfrac{1}{n+1}\right)L_n + \dfrac{1}{n+1}\left( \dfrac{1}{n}\lVert M_{n+1} \rVert + \lVert \Delta_n \rVert\right),\nonumber\\
    &\leqslant \left( 1-\dfrac{e}{n+1}\right)L_n + \dfrac{1}{n(n+1)}\Big(E\sqrt{n\log(n)}+C\log (n)+D\Big),\label{hsna_th_rates_eq13b}
\end{align}
where $e=1-c$. Then, one can show that
\begin{equation}\label{hsna_th_rates_eq14}
L_n \leqslant \left( \dfrac{2}{n+1}\right)^e L_1 + \sum_{k=2}^n \left(\dfrac{k+1}{n+1}\right)^e \dfrac{1}{k(k+1)}\Big(E\sqrt{k\log(k)}+C\log (k)+D\Big).
\end{equation}
Hence, it follows from \eqref{hsna_th_rates_eq14} that for all $n\geqslant 2$
\begin{align*}
    L_n&\leqslant \left( \dfrac{2}{n+1}\right)^e L_1 + \dfrac{A}{(n+1)^e}\sum_{k=2}^n \dfrac{\sqrt{\log (k)}}{\sqrt{k}(k+1)^{1-e}},\\
    &\leqslant \left( \dfrac{2}{n}\right)^e L_1 + \dfrac{A\sqrt{\log (n)}}{n^e}\sum_{k=2}^n \dfrac{1}{k^{a}},\\
    &\leqslant \left( \dfrac{2}{n}\right)^e L_1 + \dfrac{A\sqrt{\log (n)}}{(1-a)n^{e+a-1}} \qquad\qquad \text{a.s.},
\end{align*}
where $A=\max(C,D,E)$ and $a=3/2-e=1/2+c$. We recall that the positive constant $c$ has been chosen such that $c<1/2$ which means that $0 < a < 1$. Then, we obtain for all $n\geqslant 2$ that
\begin{equation*}
    L_n\leqslant \left( \dfrac{2}{n}\right)^e L_1 + \dfrac{A}{(1-a)} \sqrt{\dfrac{\log (n)}{n}} \qquad \text{a.s.}
\end{equation*}
Since $e>1/2$, we deduce that
\begin{equation}\label{hsna_th_rates_eq15}
    L_n^2= \mathcal{O} \left(\dfrac{\log (n)}{n}\right) \qquad \text{a.s.}
\end{equation}
Hence, we have from \eqref{hsna_th_rates_eq12} and \eqref{hsna_th_rates_eq15} that
\begin{equation}\label{hsna_th_rates_eq16}
    \lVert \Delta_n \rVert^2= \mathcal{O} \left(\dfrac{\log (n)}{n}\right) \qquad \text{a.s.}
\end{equation}
Finally, the combination of \eqref{hsna_th_rates_eq5}, \eqref{hsna_th_rates_eq8b8} and \eqref{hsna_th_rates_eq16} ensures that
\begin{equation}\label{hsna_th_rates_eq17}
    \lVert \widehat{\theta}_n-\theta \rVert^2= \mathcal{O} \left(\dfrac{\log (n)}{n}\right) \qquad \text{a.s.},
\end{equation}
which is exactly the almost sure rate of convergence \eqref{hsna_th_rates_res1}.
Hereafter, we can rewrite the matrix $\overline{S}_n$ for all $n\geqslant 1$ as
\begin{align}
    \overline{S}_n &=\dfrac{1}{n}\sum_{k=1}^n c_k \Phi_k \Phi_k^T + \dfrac{1}{n}\I_d,\nonumber\\
    &=\dfrac{1}{n} \sum_{k=1}^n \Big\{c_k \Phi_k \Phi_k^T-\E\Big[c_k\Phi_k\Phi_k^T\Big\lvert \F_{k-1} \Big]\Big \} + \dfrac{1}{n} \sum_{k=1}^n \E\Big[c_k\Phi_k\Phi_k^T\Big\lvert \F_{k-1} \Big]+ \dfrac{1}{n}\I_d,\nonumber\\
    &=\dfrac{1}{n} T_n+\dfrac{\alpha}{n}\sum_{k=1}^n\E\Big[a_k\Phi_k\Phi_k^T\Big\lvert \F_{k-1} \Big] +\dfrac{\beta}{n}\sum_{k=1}^n\E\Big[b_k\Phi_k\Phi_k^T\Big\lvert \F_{k-1} \Big]
     + \dfrac{1}{n}\I_d,\label{hsna_th_rates_eq18}
\end{align}
where
\begin{equation*}
    T_n= \sum_{k=1}^n \Big(c_k \Phi_k \Phi_k^T-\E\Big[c_k\Phi_k\Phi_k^T\Big\lvert \F_{k-1} \Big]\Big ).
\end{equation*}
By using the same technique as in equation \eqref{hsna_th_rates_eq8b2}, we have that
\begin{equation}\label{hsna_th_rates_eq19}
    \E\Big[b_k\Phi_k\Phi_k^T\Big\lvert \F_{k-1} \Big]= S+\E\Big[\Big(\pi(\widehat{\theta}_{k-1}^T\Phi_k)-\pi(\theta^T\Phi_k)\Big)^2\Phi_k\Phi_k^T\Big\lvert \F_{k-1} \Big] \qquad \text{a.s.}
\end{equation}
Moreover, we immediately obtain that
\begin{equation}\label{hsna_th_rates_eq19b}
    \E\Big[a_k\Phi_k\Phi_k^T\Big\lvert \F_{k-1} \Big]= S+\E\Big[\Big(a_k-a_k(\theta)\Big)\Phi_k\Phi_k^T\Big\lvert \F_{k-1} \Big] \qquad \text{a.s.}
\end{equation}
Hence, it follows from \eqref{hsna_th_rates_eq19} and \eqref{hsna_th_rates_eq19b} that
\begin{align}
    \overline{S}_n &= \dfrac{1}{n} T_n +(\alpha+\beta)S   +\dfrac{\alpha}{n}\sum_{k=1}^n\E\Big[(a_k-a_k(\theta))\Phi_k\Phi_k^T\Big\lvert \F_{k-1} \Big] \nonumber \\
    & \; +\dfrac{\beta}{n}\sum_{k=1}^n\E\Big[(\pi(\widehat{\theta}_{k-1}^T\Phi_k)-\pi(\theta^T\Phi_k))^2\Phi_k\Phi_k^T\Big\lvert \F_{k-1} \Big] + \dfrac{1}{n}\I_d,\label{hsna_th_rates_eq20}
\end{align}
which leads to
\begin{align}
    \lVert \overline{S}_n -(\alpha+\beta)S \rVert &\leqslant  \dfrac{1}{n} \lVert T_n\rVert +\dfrac{\alpha}{n}\sum_{k=1}^n \left \lVert \E\Big[\Big(a_k-a_k(\theta)\Big)\Phi_k\Phi_k^T\Big\lvert \F_{k-1} \Big] \right \rVert \nonumber \\
    & \; +\dfrac{\beta}{n} \sum_{k=1}^n \left \lVert \E\Big[\Big(\pi(\widehat{\theta}_k^T\Phi_k)-\pi(\theta^T\Phi_k)\Big)^2\Phi_k\Phi_k^T\Big\lvert \F_{k-1} \Big] \right \rVert + \dfrac{d}{n}.\label{hsna_th_rates_eq20b}
\end{align}
However, we have as in the inequalities \eqref{hsna_conv_ps_eq27b} and \eqref{hsna_th_rates_eq8b3b} that almost surely,
\begin{equation*}
    \left \lVert \E\Big[\Big(\pi(\widehat{\theta}_{k-1}^T\Phi_k)-\pi(\theta^T\Phi_k)\Big)^2\Phi_k\Phi_k^T\Big\lvert \F_{k-1} \Big] \right \rVert \leqslant \dfrac{d_\Phi ^4}{16} \lVert \widehat{\theta}_{k-1}-\theta\rVert^2,
\end{equation*}
and 
\begin{equation*}
    \left \lVert \E\Big[\Big(a_k-a_k(\theta)\Big)\Phi_k\Phi_k^T\Big\lvert \F_{k-1} \Big] \right \rVert \leqslant \dfrac{d_\Phi ^3}{12\sqrt{3}} \lVert \widehat{\theta}_{k-1}-\theta\rVert.
\end{equation*}
Consequently, it follows from \eqref{hsna_th_rates_eq20b} that
\begin{equation}\label{hsna_th_rates_eq21}
   \lVert \overline{S}_n -(\alpha+\beta)S \rVert \leqslant  \dfrac{1}{n} \lVert T_n\rVert + \dfrac{\alpha d_\Phi ^3}{12\sqrt{3}}\dfrac{1}{n}\sum_{k=1}^n \lVert \widehat{\theta}_{k-1}-\theta\rVert+\dfrac{\beta d_\Phi ^4}{16} \dfrac{1}{n} \sum_{k=1}^n \lVert \widehat{\theta}_{k-1}-\theta\rVert^2
  + \dfrac{d}{n}. 
\end{equation}
Finally, we deduce from \eqref{hsna_th_rates_res1} and \eqref{hsna_th_rates_eq21} that there exist finite positive random variables $E_1$ and $E_2$ such that
\begin{align}
\lVert \overline{S}_n -(\alpha+\beta)S \rVert &\leqslant  \dfrac{1}{n} \lVert T_n\rVert + \dfrac{\alpha d_\Phi^3 E_1}{12\sqrt{3}}\dfrac{1}{n}\sum_{k=1}^n \sqrt{\dfrac{\log (k)}{k}}+\dfrac{\beta d_\Phi^4 E_2}{16} \dfrac{1}{n} \sum_{k=1}^n \dfrac{\log (k)}{k}+ \dfrac{d}{n},\nonumber \\
 &\leqslant  \dfrac{1}{n} \lVert T_n\rVert 
 + \dfrac{\alpha d_\Phi^3 E_1}{6\sqrt{3}} \sqrt{\dfrac{\log(n)}{n}}+\dfrac{\beta d_\Phi^4 E_2}{16} \dfrac{\log^2 (n)}{n} + \dfrac{d}{n} \qquad \text{a.s.} \label{hsna_th_rates_eq22}
\end{align}
Moreover, by construction, the sequence $(T_n)$ is a locally square-integrable multidimensional martingale. Since the random vector $\Phi$ is bounded, we have from the strong law of large numbers for multidimensional martingales given, e.g., by Theorem 4.3.16 in \cite{duflo1996algorithmes} that
\begin{equation*}\label{hsna_th_rates_eq23}
    \lVert T_n \rVert^2=\mathcal{O}(n\log(n)) \qquad \text{a.s.}
\end{equation*}
which ensures with \eqref{hsna_th_rates_eq22} that
\begin{equation*}
    \lVert \overline{S}_{n} -(\alpha+\beta)S\rVert^2= \mathcal{O}\left( \dfrac{\log(n)}{n} \right) \qquad \text{a.s.},
\end{equation*}
The last result \eqref{hsna_th_rates_res3} is obtained by using \eqref{hsna_th_rates_res2} and the following equation
\begin{equation}
    \overline{S}_{n}^{-1} -(\alpha+\beta)^{-1}S^{-1} = \overline{S}_{n}^{-1} \Big ((\alpha+\beta)S-\overline{S}_{n} \Big)(\alpha+\beta)^{-1}S^{-1},
\end{equation}
which completes the proof of Theorem \ref{hsna_th_rates}.
\end{proof}

%%%%%%%%%%%%%%%%%%%%%%%%%%%%%%%%%%%%%%%%%%%%%%%%%%%%%%%%%%%%%%%%%%%%%%%%%

\section{Proof of Theorem \ref{hsna_th_clt}}\label{app_hsna_th_clt}

%%%%%%%%%%%%%%%%%%%%%%%%%%%%%%%%%%%%%%%%%%%%%%%%%%%%%%%%%%%%%%%%%%%%%%%%%

\begin{proof}
We recall from \eqref{hsna_th_rates_eq5} that for all $n\geqslant 0$,
\begin{equation}\label{hsna_th_clt_eq1}
  \widehat{\theta}_{n+1}-\theta=-\dfrac{1}{n+1}M_{n+1}-\Delta_n
\end{equation}
which leads to
\begin{equation}\label{hsna_th_clt_eq2}
  \sqrt{n+1}(\widehat{\theta}_{n+1}-\theta)=-\dfrac{1}{\sqrt{n+1}}M_{n+1}-\sqrt{n+1}\Delta_n.
\end{equation}
where $\sqrt{n+1}\Delta_n= O_n+P_n+Q_n$ and 
\begin{align*}
    O_n &=\dfrac{1}{\sqrt{n+1}} \sum_{k=0}^n \Big(\overline{S}_{k+1}^{-1}-\overline{S}_{k}^{-1}\Big)Z_{k+1}, \\
    P_n&=\dfrac{1}{\sqrt{n+1}}\sum_{k=0}^n\Big(\overline{S}_{k}^{-1}-S^{-1}\Big)S(\widehat{\theta}_k-\theta),\\
    Q_n&=\dfrac{1}{\sqrt{n+1}}\sum_{k=0}^n \overline{S}_{k}^{-1}\delta_k.
\end{align*}
It immediately follows from \eqref{hsna_th_rates_eq11} that
\begin{equation*}
    \lVert O_n \rVert =\mathcal{O}\left( \dfrac{\log(n)}{\sqrt{n}}\right) \qquad \text{a.s.},
\end{equation*}
which ensures that
\begin{equation}\label{hsna_th_clt_eq3}
    \lim_{n\to +\infty} \lVert O_n \rVert = 0 \qquad \text{a.s.}
\end{equation}
Moreover, we deduce from \eqref{hsna_th_rates_res1} and \eqref{hsna_th_rates_res2} that
\begin{equation*}
    \lVert P_n \rVert= \mathcal{O}\left(  \dfrac{1}{\sqrt{n}} \sum_{k=1}^n \dfrac{\log(k)}{k}\right)=\mathcal{O}\left(\dfrac{\log^2(n)}{\sqrt{n}} \right) \qquad \text{a.s.},
\end{equation*}
which implies that
\begin{equation}\label{hsna_th_clt_eq4}
    \lim_{n\to +\infty} \lVert P_n \rVert = 0 \qquad \text{a.s.}
\end{equation}
Furthermore, we have from  \eqref{hsna_conv_ps_eq27b} and \eqref{hsna_th_rates_eq4} that
\begin{align}
    \lVert \delta_n \rVert &\leqslant \left\lVert\int_0^1 \nabla^2 G(\theta + t(\widehat{\theta}_{n}-\theta))(\widehat{\theta}_{n}-\theta) \text{dt} -\nabla^2 G(\theta) (\widehat{\theta}_{n}-\theta)\right\rVert \nonumber\\
   & \leqslant \int_0^1 \left\lVert \nabla^2 G(\theta + t(\widehat{\theta}_{n}-\theta))-\nabla^2 G(\theta)\right\rVert \lVert\widehat{\theta}_{n}-\theta\rVert \text{dt}\nonumber \\
   &\leqslant \dfrac{d_\Phi^3}{24\sqrt{3}} \lVert\widehat{\theta}_{n}-\theta\rVert^2.\label{hsna_th_clt_eq4b}
\end{align}
Consequently, we obtain from \eqref{hsna_th_rates_res1} that
\begin{equation}\label{hsna_th_clt_eq5}
    \lVert Q_n \rVert =\mathcal{O}\left( \dfrac{1}{\sqrt{n}}\sum_{k=1}^n \lVert \widehat{\theta}_{k}-\theta\rVert^2 \right)=\mathcal{O}\left( \dfrac{1}{\sqrt{n}}\sum_{k=1}^n \dfrac{\log(k)}{k} \right)=\mathcal{O}\left( \dfrac{\log^2(n)}{\sqrt{n}}\right) \qquad \text{a.s.}
\end{equation}
which leads to
\begin{equation}\label{hsna_th_clt_eq6}
    \lim_{n\to +\infty} \lVert Q_n \rVert = 0 \qquad \text{a.s.}
\end{equation}
By combining equations \eqref{hsna_th_clt_eq3}, \eqref{hsna_th_clt_eq4} and \eqref{hsna_th_clt_eq6}, we deduce that
\begin{equation}\label{hsna_th_clt_eq7}
    \lim_{n\to +\infty} \sqrt{n+1}\lVert \Delta_n\rVert = 0 \qquad \text{a.s.}
\end{equation}
Hence, it only remains to study the asymptotic behavior of the martingale $(M_n)$. We already proved in \eqref{hsna_th_rates_eq8b7} that its predictable quadratic variation satisfies
\begin{equation}\label{hsna_th_clt_eq8}
        \lim_{n\to +\infty} \dfrac{1}{n}\langle M\rangle_n= \left(\nabla^2 G(\theta)\right)^{-1} \qquad \text{a.s.}
\end{equation}
Once again, since the random vector $\Phi$ is bounded, we have
\begin{equation*}%\label{hsna_th_clt_eq9}
    \lVert \varepsilon_{n+1} \rVert =\lVert Z_{n+1}-\nabla G(\widehat{\theta}_{n})\rVert \leqslant \lVert \Phi_{n+1} \rVert + \E[\lVert \Phi \rVert]\leqslant 2 d_\Phi \qquad \text{a.s.}
\end{equation*}
which leads to 
\begin{equation}\label{hsna_th_clt_eq10}
    \underset{n\geqslant 1}{\sup} \, \E \big[\lVert \varepsilon_{n+1} \rVert^4 \big\lvert \F_{n} \big] \leq 16 d_\Phi^4 \qquad \text{a.s.}
\end{equation}
Consequently, the sequence $(M_n)$ satisfies Lindeberg’s condition. Therefore, we deduce from the central limit theorem for martingales given by Corollary 2.1.10 in \cite{duflo1996algorithmes} that
\begin{equation*}\label{hsna_th_clt_eq11}
    \dfrac{1}{\sqrt{n}} M_{n} \quad \overset{\mathcal{L}}{\underset{n\to +\infty}{\longrightarrow}} \quad \mathcal{N}_d\left(0,\left(\nabla^2 G(\theta)\right)^{-1}\right),
\end{equation*}
which achieves the proof of \eqref{hsna_th_clt_res1} thanks to \eqref{hsna_th_clt_eq2} and \eqref{hsna_th_clt_eq7}.
\end{proof}

%%%%%%%%%%%%%%%%%%%%%%%%%%%%%%%%%%%%%%%%%%%%%%%%%%%%%%%%%%%%%%%%%%%%%%%%%

\section{Proof of Theorem \ref{hsna_th_qsl}}\label{app_hsna_th_qsl}

%%%%%%%%%%%%%%%%%%%%%%%%%%%%%%%%%%%%%%%%%%%%%%%%%%%%%%%%%%%%%%%%%%%%%%%%%

\begin{proof}
We recall from \eqref{hsna_th_rates_eq5} that for all $n\geqslant 1$,
\begin{equation*}\label{hsna_th_qsl_eq1}
  \widehat{\theta}_{n}-\theta=-\dfrac{1}{n}M_{n}-\Delta_{n-1}
\end{equation*}
which implies that
\begin{equation}\label{hsna_th_qsl_eq2}
     (\widehat{\theta}_{n}-\theta) (\widehat{\theta}_{n}-\theta)^T =\dfrac{1}{n^2}  M_{n}M_{n}^T +\Delta_{n-1}\Delta_{n-1}^T+\dfrac{1}{n} \Big( M_{n}\Delta_{n-1}^T +
     \Delta_{n-1}M_n^T\Big).
\end{equation}
By summing on both sides of \eqref{hsna_th_qsl_eq2}, we obtain that
\begin{equation}
  \dfrac{1}{\log(n)}\sum_{k=1}^n ( \widehat{\theta}_{k}-\theta) (\widehat{\theta}_{k}-\theta)^T =\dfrac{1}{\log(n)}\sum_{k=1}^n \dfrac{1}{k^2} M_{k}M_{k}^T + \dfrac{1}{\log(n)}\sum_{k=1}^n \Delta_{k-1}\Delta_{k-1}^T + N_n\label{hsna_th_qsl_eq3}
\end{equation}
where
\begin{equation*}
    N_n= \dfrac{1}{\log(n)}\sum_{k=1}^n \dfrac{1}{k} \Big( M_{k}\Delta_{k-1}^T +
     \Delta_{k-1}M_k^T\Big).
\end{equation*}
We already saw from \eqref{hsna_th_rates_eq8b7} that 
\begin{equation}\label{hsna_th_qsl_eq4}
        \lim_{n\to +\infty} \dfrac{1}{n}\langle M\rangle_n= \left(\nabla^2 G(\theta)\right)^{-1} \qquad \text{a.s.}
\end{equation}
Moreover, we have already checked that the martingale $(M_n)$ satisfied Lindeberg’s condition. Furthermore, it follows from \eqref{hsna_th_clt_eq10} that
\begin{equation*}
    \sum_{n=2}^\infty \frac{1}{(n\log(n))^2} 
    \E \big [ \lVert\overline{S}_{n}^{-1}\varepsilon_{n+1} \rVert^4 \big\lvert \F_{n} \big] < \infty \qquad \text{a.s.}
\end{equation*}
Consequently, we deduce from the quadratic strong law for multi-dimensional martingale given e.g. by Theorem A.2 in \cite{bercu2021center} that
\begin{equation}\label{hsna_th_qsl_eq5}
    \lim_{n\to +\infty} \dfrac{1}{\log(n)}\sum_{k=1}^n \dfrac{1}{k^2} M_{k}M_{k}^T = \left(\nabla^2 G(\theta)\right)^{-1} \qquad \text{a.s.}
\end{equation}
Hereafter, the almost sure rate of convergence of $(\Delta_n)$ established in \eqref{hsna_th_rates_eq16} is not enough to properly handle the second term in the right-hand side of \eqref{hsna_th_qsl_eq3}. We need to be more precise on that rate by exploiting the upper-bound in \eqref{hsna_th_clt_eq4b}. For that end, we have from \eqref{hsna_th_rates_eq7b} and \eqref{hsna_th_rates_eq11} combined with the almost sure rates of convergence \eqref{hsna_th_rates_res1} and \eqref{hsna_th_rates_res3} that
\begin{equation}\label{hsna_th_qsl_eq6}
    \lVert \Delta_n \rVert = \mathcal{O}\left( \dfrac{\log^2(n)}{n}\right) \qquad \text{a.s.},
\end{equation}
which leads to 
\begin{equation}\label{hsna_th_qsl_eq7}
    \sum_{k=1}^n \lVert \Delta_k \rVert^2 = \mathcal{O}(1)+ \mathcal{O}\left( \sum_{k=1}^n \dfrac{\log^4(k)}{k^2}\right)= \mathcal{O}(1) \qquad \text{a.s.}
\end{equation}
Hence, we immediately deduce that
\begin{equation}\label{hsna_th_qsl_eq8}
     \lim_{n\to +\infty} \dfrac{1}{\log(n)} \sum_{k=1}^n \lVert \Delta_k \rVert^2 = 0 \qquad \text{a.s.}
\end{equation}
By using twice the Cauchy-Schwarz inequality, we also obtain from 
\eqref{hsna_th_qsl_eq5} and \eqref{hsna_th_qsl_eq8} that 
\begin{equation}\label{hsna_th_qsl_eq10}
    \lim_{n\to +\infty} N_n= 0 \qquad \text{a.s.}
\end{equation}
Therefore, it follows from \eqref{hsna_th_qsl_eq3}, \eqref{hsna_th_qsl_eq5}, \eqref{hsna_th_qsl_eq8} and \eqref{hsna_th_qsl_eq10} that
\begin{equation}\label{hsna_th_qsl_eq11}
    \lim_{n\to +\infty} \dfrac{1}{\log(n)}\sum_{k=1}^n ( \widehat{\theta}_{k}-\theta) (\widehat{\theta}_{k}-\theta)^T = \left(\nabla^2 G(\theta)\right)^{-1} \qquad \text{a.s.}
\end{equation}
Furthermore, we have from the Taylor expansion of $G$ with integral remainder that
\begin{equation*}
    G(\widehat{\theta}_{n})-G(\theta)= \dfrac{1}{2}(\widehat{\theta}_{n}-\theta)^T\nabla^2 G(\theta)(\widehat{\theta}_{n}-\theta) + I_n,
\end{equation*}
where
\begin{equation}\label{hsna_th_qsl_eq11b}
    I_n= \dfrac{1}{2}\int_0^1 (\widehat{\theta}_{n}-\theta)^T \Big(\nabla^2 G(\theta + t(\widehat{\theta}_{n}-\theta)) -\nabla^2 G(\theta) \Big) (\widehat{\theta}_{n}-\theta) \text{dt}.
\end{equation}
Then, we obtain that 
\begin{equation}\label{hsna_th_qsl_eq12}
    \dfrac{1}{\log(n)} \sum_{k=1}^n \Big(G(\widehat{\theta}_{k})-G(\theta)\Big)
    =\dfrac{1}{2\log(n)} \sum_{k=1}^n (\widehat{\theta}_{k}-\theta)^T\nabla^2 G(\theta)(\widehat{\theta}_{k}-\theta) +\dfrac{1}{\log(n)} \sum_{k=1}^n I_k.
\end{equation}
However, one can easily see that
\begin{align}
    &\dfrac{1}{\log(n)} \sum_{k=1}^n (\widehat{\theta}_{k}-\theta)^T\nabla^2 G(\theta)(\widehat{\theta}_{k}-\theta)\nonumber\\
    &=\dfrac{1}{\log(n)} \sum_{k=1}^n \Tr\Big((\widehat{\theta}_{k}-\theta)^T\nabla^2 G(\theta)(\widehat{\theta}_{k}-\theta)\Big) \nonumber\\
    &=\dfrac{1}{\log(n)} \sum_{k=1}^n \Tr\Big( \left(\nabla^2 G(\theta)\right)^{1/2}(\widehat{\theta}_{k}-\theta)(\widehat{\theta}_{k}-\theta)^T \left(\nabla^2 G(\theta)\right)^{1/2}\Big) \nonumber\\
    &=\Tr\left( \left(\nabla^2 G(\theta)\right)^{1/2} \left(\dfrac{1}{\log(n)} \sum_{k=1}^n (\widehat{\theta}_{k}-\theta)(\widehat{\theta}_{k}-\theta)^T \right)\left(\nabla^2 G(\theta)\right)^{1/2}\right). \label{hsna_th_qsl_eq13}
\end{align}
Consequently, we deduce from \eqref{hsna_th_qsl_eq11} that
\begin{equation} \label{hsna_th_qsl_eq14}
    \lim_{n\to +\infty} \dfrac{1}{\log(n)} \sum_{k=1}^n (\widehat{\theta}_{k}-\theta)^T\nabla^2 G(\theta)(\widehat{\theta}_{k}-\theta)= d \qquad \text{a.s.}
\end{equation}
Moreover, as in equation \eqref{hsna_th_clt_eq4b}, we obtain from \eqref{hsna_th_qsl_eq11b} that
\begin{equation} \label{hsna_th_qsl_eq14b}
    \lvert I_n \rvert \leqslant \dfrac{d_\Phi^3}{48\sqrt{3}} \lVert\widehat{\theta}_{n}-\theta\rVert^3,
\end{equation}
which implies that
\begin{equation}\label{hsna_th_qsl_eq15}
    \left\lvert \dfrac{1}{\log(n)} \sum_{k=1}^n I_k \right \rvert
    \leqslant \dfrac{d_\Phi^3}{48\sqrt{3}} \dfrac{1}{\log(n)} \sum_{k=1}^n \lVert \widehat{\theta}_{k}-\theta \rVert^3.
\end{equation}
Nevertheless, the almost sure rate of convergence $(\widehat{\theta}_{n})$ given by \eqref{hsna_th_rates_res1} implies that
\begin{equation*}
   \sum_{k=1}^n \lVert \widehat{\theta}_{k}-\theta \rVert^3=  \mathcal{O}(1)+ \mathcal{O}\left( \sum_{k=1}^n \dfrac{\log^{3/2}(k)}{k^{3/2}}\right)= \mathcal{O}(1) \qquad \text{a.s.}
\end{equation*}
which immediately leads to
\begin{equation} \label{hsna_th_qsl_eq16}
    \lim_{n\to +\infty} \dfrac{1}{\log(n)}\sum_{k=1}^n \lVert \widehat{\theta}_{k}-\theta \rVert^3= 0 \qquad \text{a.s.}
\end{equation}
Therefore, we deduce from the combination of \eqref{hsna_th_qsl_eq15} and \eqref{hsna_th_qsl_eq16}
that
\begin{equation}
     \lim_{n\to +\infty} \dfrac{1}{\log(n)} \sum_{k=1}^n I_k = 0\qquad \text{a.s.} \label{hsna_th_qsl_eq17}
\end{equation}
Finally, we conclude from \eqref{hsna_th_qsl_eq12}, \eqref{hsna_th_qsl_eq14} and \eqref{hsna_th_qsl_eq17} that
\begin{equation*}
    \lim_{n\to +\infty} \dfrac{1}{\log(n)} \sum_{k=1}^n \Big(G(\widehat{\theta}_{k})-G(\theta)\Big)= \dfrac{d}{2} \qquad \text{a.s.}
\end{equation*}
which completes the proof of Theorem \ref{hsna_th_qsl}.
\end{proof}
\end{appendix}

%%%%%%%%%%%%%%%%%%%%%%%%%%%%%%%%%%%%%%%%%%%%%%
%% Acknowledgements                         %%
%% should be provided in the                %%
%% Acknowledgements section.                %%
%%%%%%%%%%%%%%%%%%%%%%%%%%%%%%%%%%%%%%%%%%%%%%
% \begin{acks}[Acknowledgments]
% %(to be completed)
% The authors would like to thank the anonymous referees, an Associate
% Editor and the Editor for their constructive comments that improved the
% quality of this paper.
% \end{acks}

%%%%%%%%%%%%%%%%%%%%%%%%%%%%%%%%%%%%%%%%%%%%%%
%% Funding information, if any,             %%
%% should be provided in the                %%
%% funding section.                         %%
%%%%%%%%%%%%%%%%%%%%%%%%%%%%%%%%%%%%%%%%%%%%%%
\begin{funding}
This project has benefited from state support managed by the Agence Nationale de la Recherche (French National Research Agency) under the reference ANR-20-SFRI-0001. 
\end{funding}

\bibliographystyle{imsart-number} % Style BST file (imsart-number.bst or imsart-nameyear.bst)
%\bibliography{bibliography}       % Bibliography file (usually '*.bib')

\end{document}